\documentclass[11pt]{article}
\usepackage[margin=1in]{geometry}
\usepackage[T1]{fontenc}
\usepackage{lmodern}
\usepackage{microtype}
\usepackage{float}      % For H placement
\usepackage{placeins} 
\usepackage{booktabs}
\usepackage{amsmath,amsthm,amssymb}
\usepackage{graphicx}
\usepackage{color}
\usepackage{enumerate}
\usepackage{multicol}
\usepackage{tikz}
\usepackage[utf8]{inputenc}
\usepackage{mathrsfs}
\usepackage[
backend=biber,
style=numeric-comp,
sorting=nyt,
maxbibnames=99,
giveninits=true,
doi=true,
url=false,
isbn=false
]{biblatex}
\usepackage{titlesec}
\usepackage{enumitem}
\usepackage{longtable}
\titlespacing*{\section}{0pt}{2.5ex plus 0.8ex minus 0.2ex}{1.2ex}
\titlespacing*{\subsection}{0pt}{2.0ex plus 0.6ex minus 0.2ex}{0.8ex}

\setlist[enumerate]{topsep=0pt, itemsep=2pt, parsep=0pt, partopsep=0pt}
\setlist[itemize]{itemsep=1pt, topsep=0pt, parsep=0pt}
\allowdisplaybreaks[2]

\theoremstyle{definition}
\newtheorem{axiom}{Axiom}
\numberwithin{axiom}{section}

\AddToHook{env/incorrectaxiom/after}{\addtocounter{axiom}{-1}}
\newtheorem{definition}{Definition}
\theoremstyle{plain}
\newtheorem{theorem}{Theorem}
\newtheorem{lemma}{Lemma}
\newtheorem{proposition}{Proposition}

\theoremstyle{definition}

\theoremstyle{remark}

\DeclareMathOperator{\supp}{supp}

\usepackage[hidelinks]{hyperref} 
\usepackage{cleveref}                 % \Cref{thm:...} 自动输出 Theorem/Lemma...
\usepackage{environ}                  % 捕获环境 BODY
\usepackage{etoolbox}                 % 

\newif\ifdeferproofs
\deferproofstrue   
\newcounter{appproof}[section]
\renewcommand{\theappproof}{\thesection.\arabic{appproof}}

\newcommand{\AllDeferredProofs}{}

\NewEnviron{deferredproof}[1]{%
\ifdeferproofs
\par\noindent\emph{Proof.}%
\hyperref[prf:#1]{See the proof in Appendix~\ref*{prf:#1}}%
\par\medskip
\global\edef\AllDeferredProofs{%
\unexpanded\expandafter{\AllDeferredProofs}%
\unexpanded{%
\par\bigskip
\refstepcounter{appproof}\label{prf:#1}%
\subsection*{Proof~\theappproof\; (Proof of \Cref{#1})}%
\noindent\hyperref[#1]{Back to \Cref*{#1}.}\par\medskip
\begin{proof}
}
\expandafter\unexpanded\expandafter{\BODY}%
\unexpanded{%
\end{proof}
}%
}%
\else
\begin{proof}[Proof of \Cref{#1}]
\BODY
\end{proof}
\fi
}

\author{Wei ShuAng}

\title{A Proof for Generalized Egalitarianism}

\AtBeginDocument{%
\setlength{\abovedisplayskip}{6pt plus 2pt minus 2pt}%
\setlength{\belowdisplayskip}{6pt plus 2pt minus 2pt}%
\setlength{\abovedisplayshortskip}{4pt plus 1pt minus 1pt}%
\setlength{\belowdisplayshortskip}{4pt plus 1pt minus 1pt}%
}

\hypersetup{pdftitle={A Proof for Generalized Egalitarianism},pdfauthor={Wei ShuAng},pdfsubject={Second-variation revision of Lorenz-area rigidity}}

\begin{document}

\maketitle

\begin{abstract}
This paper proves a rigidity theorem characterizing the Gini-type Lorenz-area ordering under general decision-theoretic constraints together with weak Lorenz dominance, thereby resolving the characterization problem pursued in this line of research on egalitarianism. This does not, of course, bring the broader inquiry to an end. In a certain sense, it may instead mark a step beyond the infancy of our attempt to connect social choice theory, welfare economics, and formal moral philosophy through a common mathematical language.
\end{abstract}

\clearpage
\setcounter{tocdepth}{2}
\tableofcontents
\clearpage

\section{Introduction and the main theorem}

This paper studies the formal representation of egalitarian evaluation. It will be shown that, under the constraints imposed by structural axioms of decision theory, with convex curves as the objects of ordering, a generalized Gini functional can be deduced as an affine representation, unique up to positive affine transformation. It's important to note that this paper does not claim to derive anything normative from non-moral premises: the egalitarian moral judgment is already contained in the simple weak dominance axiom. What is interesting is the fact that together with the remaining axioms, this moral premise yields, through mathematical deduction, a Gini-type representation that appears considerably stronger. 

Indeed, the Lorenz curve has been used as a standard device for representing relative distributions and measuring inequality \parencite{arnoldSarabia2018}. Like many useful formulas throughout history, Lorenz curves and Gini coefficients have an intuitive appeal that does not by itself provide an axiomatic justification. However, due to their intuitive usefulness---they do seem to capture many profound features of inequality---these two concepts have gained widespread acceptance, even though, from a somewhat pedantic perspective, their intuitive appeal alone does not explain why they seem so natural, and the choice of relevant properties could be deemed arbitrary. Two broad approaches to this problem can be distinguished. The first
one, developed in work by Kolm and by Blackorby and Donaldson,
starts with axioms that encode explicit normative evaluations and
seeks representations of inequality and aggregate social welfare functions
\parencite{kolm1976a,blackorbyDonaldson1978,donaldsonWeymark1980,weymark1981}.
The other line of research, instead, tries to integrate the study
of distributive ethics with the classical framework of
decision-theoretic rationality with minimal normative input \parencite{aaberge2001,yaari1988}.
To provide a more vivid picture, under the current framing, an abstract domain of curves corresponding to patterns of distribution of a relevant quantity, such as money or individual utility, will be the basic space of investigation. Indeed, the basic idea behind the following proof is to model egalitarian evaluation as a choice among normalized Lorenz curves rather than among outcomes or lotteries in the ordinary EUT framework.

This line of ingenious reasoning, which treats objects other than abstractions of first-order outcomes or lotteries as objects of investigation using decision-theoretic tools, has an analytical parallel in Arrow–Pratt's much-celebrated approach to quantifying risk aversion. Indeed, in the pioneering works \parencite{pratt1964,arrow1965}, the relevant results investigate the risk-bearing properties of individual preferences through the curvature of their utility functions. More specifically, this approach formulates risk attitudes directly in terms of preferences over non-certain monetary gambles, assuming expected utility. For \(u,v\in C^2(I)\) with \(u',v'>0\) on a common open wealth interval \(I\), the inequality \(-v''/v'\ge -u''/u'\) throughout \(I\) is equivalent to \(v(x)=f(u(x))\) for an increasing concave function \(f\). Equivalently, \(v\) assigns a weakly lower certainty equivalent to every finitely supported lottery taking values in \(I\). This line of thought has been systematically developed in modern quantitative finance; see \textcite[Chapter~2]{gollier2001} and \textcite[Section~2]{follmerSchied2016}. The curvature comparison is invariant under positive affine changes of either utility index.
The research approach most directly relevant here, treating Lorenz curves as the basic objects of analysis, can be found in Aaberge~\cite{aaberge2001}. He adopted the standard decision-theoretic requirements of completeness, mixture independence, continuity, weak dominance, and nontriviality \parencite{hersteinMilnor1953,fishburn1970}, together with a dual independence condition defined through mixtures in the inverse chart as axiomatic conditions for the representation results.

Aaberge's proposed route to the proof of Lorenz area ordering proceeds through two separate one-sided primal and dual integral-representation theorems. Unfortunately, in our earlier paper \textit{Measure of Morality} \parencite{wei2026measuremoralitymathematicaltheory}, we provided two counterexamples, one to each of the one-sided representation theorems proposed by Aaberge in his original papers. Each one-sided axiom package admits endpoint singularities. The formation of singularities block the existence of finite integral-kernel representations assumed in the original proof strategy. Therefore they invalidate the proposed proof strategies in \cite{aaberge2001}.

Our earlier paper provided a new proof of the representation by introducing a strengthened regularity condition. However, repairing the integral representation by introducing the unmotivated metric language of an additional axiom makes the repair look very ad hoc in contrast with the rest of the hypotheses. This alien character gives us some uneasiness.

As a result, the present paper will establish the Lorenz-ordering representation without the additional quantitative Archimedean hypotheses. This paper is structured as follows: the rest of this section introduces the background information regarding the axioms, the domain, and the auxiliary functions that will be used.

Section two investigates a specific attempt to construct an ordering that resembles the construction of a one-sided counterexample under the combined constraints of both primal and dual independence, and then explains why such a construction is impossible. The third and fourth sections are devoted to the proof of the main theorem itself.

\subsection{Domain, axioms, and the main theorem}
\label{sec:axioms-main-theorem}

\begin{definition}[Lorenz domain and endpoint conditions]
\label{def:lorenz-domain}
Equip \(C([0,1])\) with \(\|f\|_\infty:=\max_{0\le x\le1}|f(x)|\). Define
\[
\begin{aligned}
\mathcal{L}:=\{L\in C([0,1]):\;&L(0)=0,\quad L(1)=1,\\
&L(x)\le L(y)\quad(0\le x\le y\le1),\\
&L(\lambda x+(1-\lambda)y)\le\lambda L(x)+(1-\lambda)L(y)\\[-2pt]
&\hspace{7em}(x,y\in[0,1],\ \lambda\in[0,1])\},\\
\mathcal{L}^\uparrow:=\{L\in\mathcal{L}:\;&L(x)<L(y)\quad(0\le x<y\le1)\}.
\end{aligned}
\]

The equality curve is \(E(x):=x\). The area and its Gini normalization are
\(A(L):=\int_0^1 L(x)\,dx\) and \(J(L):=1-2A(L)\).
\end{definition}

\begin{definition}[Inverse chart and the two mixtures]
\label{def:mixtures}
Let
\[
Y:=\{F\in C([0,1]):F(0)=0,\ F(1)=1,\ F\text{ is strictly increasing and concave}\}.
\]

Inversion means the ordinary inverse of a continuous, strictly increasing bijection of \([0,1]\).
For \(L,M\in\mathcal{L}\), their primal mixture is
\(\lambda L+(1-\lambda)M\), where \(\lambda\in[0,1]\).
For \(L,M\in\mathcal{L}^\uparrow\), their dual mixture is
\[
L\oplus_\lambda M:=\bigl(\lambda L^{-1}+(1-\lambda)M^{-1}\bigr)^{-1}
\qquad(\lambda\in[0,1]).
\]

Both operations take values in their stated domains; inversion identifies \(\mathcal{L}^\uparrow\) with \(Y\), as proved in Lemma~\ref{lem:domain-inversion}.
No dual mixture involving a noninvertible curve is used.
\end{definition}

For a binary relation \(\succsim\), write \(L\sim M\) when both \(L\succsim M\) and \(M\succsim L\),
and write \(L\succ M\) when \(L\succsim M\) but not \(M\succsim L\).
A real-valued function \(V\) represents this relation when
\(L\succsim M\iff V(L)\ge V(M)\).
An affine function on a convex domain \(X\) is a function \(V\) satisfying
\(V(\lambda x+(1-\lambda)y)=\lambda V(x)+(1-\lambda)V(y)\)
for every \(x,y\in X\) and \(\lambda\in[0,1]\).

\begin{axiom}[Complete preorder]
\label{ax:weak-order}
For every \(L,M\in\mathcal{L}\), either \(L\succsim M\) or \(M\succsim L\).
For every \(L,M,N\in\mathcal{L}\),
\(L\succsim M\) and \(M\succsim N\) imply \(L\succsim N\).
\end{axiom}

\begin{axiom}[Primal mixture independence]
\label{ax:primal-independence}
For every \(L,M,N\in\mathcal{L}\) and \(\lambda\in(0,1)\),
\[
L\succsim M\quad\Longleftrightarrow\quad
\lambda L+(1-\lambda)N\succsim\lambda M+(1-\lambda)N.
\]
\end{axiom}

\begin{axiom}[Continuity in the uniform topology]
\label{ax:continuity}
The graph of \(\succsim\) is closed in \(\mathcal{L}\times\mathcal{L}\), with the product of the relative uniform topologies.
Explicitly, if \(L_n,M_n,L,M\in\mathcal{L}\), \(L_n\succsim M_n\) for every \(n\), and
\(\|L_n-L\|_\infty+\|M_n-M\|_\infty\to0\), then \(L\succsim M\).
This is the continuity hypothesis used below, not merely continuity along each mixture segment.
\end{axiom}

\begin{axiom}[Weak Lorenz dominance]
\label{ax:dominance}
For every \(L,M\in\mathcal{L}\),
\(L(x)\ge M(x)\) for all \(x\in[0,1]\) implies \(L\succsim M\).
\end{axiom}

\begin{axiom}[Nontriviality on the stated domains]
\label{ax:nontriviality}
There exist \(L,M\in\mathcal{L}\) with \(L\succ M\).
The restriction to \(\mathcal{L}^\uparrow\) is also nontrivial: there exist \(L^+,M^+\in\mathcal{L}^\uparrow\) with \(L^+\succ M^+\).
\end{axiom}

\begin{axiom}[Dual mixture independence]
\label{ax:dual-independence}
For every \(L,M,N\in\mathcal{L}^\uparrow\) and \(\lambda\in(0,1)\),
\[
L\succsim M\quad\Longleftrightarrow\quad
L\oplus_\lambda N\succsim M\oplus_\lambda N.
\]
\end{axiom}

\begin{theorem}[Rigid characterization of Lorenz-curve ordering from primal and dual mixture structures]
\label{thm:direct-gini}
Let \(\succsim\) be a binary relation on the domain of Definition~\ref{def:lorenz-domain} satisfying Axioms~\ref{ax:weak-order}--\ref{ax:dual-independence}, with the mixtures and continuity understood as defined above.
Then there exist constants \(c_0\in\mathbb R\) and \(c_1>0\) such that
\[
L_1\succsim L_2
\quad\Longleftrightarrow\quad
c_0+c_1A(L_1)\ge c_0+c_1A(L_2)
\qquad(L_1,L_2\in\mathcal{L}).
\]

Equivalently, \(\succsim\) is the Lorenz-area order, or the order represented with the opposite inequality by \(J(L)=1-2A(L)\).
\end{theorem}

The proof of the main Theorem~\ref{thm:direct-gini} is deferred to Sections~\ref{sec:proof-main} and~\ref{sec:proof-classification}.
Before proceeding to the proof proper, we examine a natural attempt to construct a counterexample and explain why it fails. The construction succeeds in satisfying the primal requirements, but the additional dual independence condition forces the ordering back to the area order. Understanding this failure explains why comparing the two mixture structures is useful.

\section{Singular weighted-integral obstruction}
\label{sec:weighted-obstruction}

This section will investigate a specific attempt to construct an ordering that
resembles the one-sided counterexample. The attempt starts with a positive weight that is nonintegrable near
zero, while \(\int_0^1u\,w(u)\,du\) remains finite.
We will verify that its weighted functional satisfies all the primal
requirements, but has no global Lipschitz bound in the supremum norm.
Affine uniqueness then rules out the Lorenz-area ordering, whose
affine representatives necessarily satisfy such a bound. The question is whether this failure of the area
representation can persist when dual independence is imposed as well.
The calculation will show that this additional requirement forces the
weight to be constant and therefore rules out the singular candidates.
\subsection{Singular weights satisfying the primal axioms}
Singular positive weights can define orderings that satisfy all the
primal axioms while differing from the Lorenz-area order.

Consider the following family:
\[
w_p(u):=u^{-p},\qquad 1\le p<2,
\]
These weights satisfy
\[
\int_0^1w_p(u)\,du=\infty,\qquad
\int_0^1u\,w_p(u)\,du=\frac1{2-p}<\infty.
\]
In this subsection, fix a weight satisfying
\[
w\in C^1((0,1)),\qquad w(u)>0\quad(0<u<1),
\]
and
\[
\int_0^1w(u)\,du=\infty,\qquad
\int_0^1u\,w(u)\,du<\infty.
\]
Define
\[
U_w(L):=\int_0^1w(u)L(u)\,du,
\qquad
L\succsim M\quad\Longleftrightarrow\quad U_w(L)\ge U_w(M).
\]
These conditions allow \(w\) to grow without bound near zero.
Every Lorenz curve satisfies \(0\le L(u)\le u\), so the integrable
function \(u\,w(u)\) controls the product \(w(u)L(u)\).
This makes the evaluation finite although \(w\) itself is not
integrable. After verifying the primal axioms, we show that
nonintegrability of \(w\) prevents \(U_w\) from satisfying a global
Lipschitz bound. This will rule out the area representation.

It's easy to verify that the induced ordering satisfies the one-sided primal axioms.
Put \(C_w=\int_0^1u\,w(u)\,du<\infty\).
Since \(0\le L(u)\le u\), one has \(0\le U_w(L)\le C_w\).
For \(L,M\in\mathcal{L}\), put \(d=\|L-M\|_\infty\).
The pointwise bounds give
\[
\begin{aligned}
|L(u)-M(u)|&\le\min\{u,d\},\\
|U_w(L)-U_w(M)|&\le\omega_w(d),\qquad
\omega_w(d):=\int_0^1\min\{u,d\}w(u)\,du.
\end{aligned}
\]
Here \(0\le\min\{u,d\}w(u)\le u\,w(u)\), and the upper
bound is integrable. Dominated convergence therefore gives
\(\omega_w(d)\to0\) as \(d\downarrow0\).
Thus \(U_w\) is continuous on the stated domain.
Linearity of the integral makes \(U_w\) affine under primal mixtures.
In particular, for \(L,M,N\in\mathcal{L}\) and \(0<\lambda<1\),
\[
\begin{aligned}
&U_w(\lambda L+(1-\lambda)N)
-U_w(\lambda M+(1-\lambda)N)\\
&\qquad=\lambda\bigl(U_w(L)-U_w(M)\bigr).
\end{aligned}
\]
Since \(\lambda>0\), the two differences have the same sign. This
is precisely primal mixture independence. If \(L\ge M\) pointwise,
positivity of \(w\) also gives
\[
U_w(L)-U_w(M)=\int_0^1w(u)(L(u)-M(u))\,du\ge0,
\]
which proves weak Lorenz dominance.
For nontriviality, take any \(Q\in\mathcal{L}^\uparrow\) with \(Q\ne E\).
For example, \(Q(u)=u^2\) has these properties.
The continuous gap \(E-Q\) is nonnegative and positive somewhere,
so there is a compact interval \([a,b]\subset(0,1)\), with \(a<b\),
on which it is strictly positive. Hence
\[
\begin{aligned}
U_w(E)-U_w(Q)
&=\int_0^1w(u)(u-Q(u))\,du\\
&\ge\left(\min_{u\in[a,b]}w(u)\right)
\int_a^b(u-Q(u))\,du>0.
\end{aligned}
\]
Therefore the induced preorder is nontrivial on both domains.
For any continuous real-valued representative, its weak-order graph is closed because the defining difference is a continuous function; this verifies Axiom~\ref{ax:continuity} here.
The numerical ordering is automatically complete and transitive, so
all of Axioms~\ref{ax:weak-order}--\ref{ax:nontriviality} have now
been verified.

If \(U_w\) and \(A\) represented the same ordering, their
primal affinity and nonconstancy would give, by affine uniqueness,
\[
U_w(L)=c_0+c_1A(L)
\qquad(L\in\mathcal{L}),\qquad c_1>0.
\]
It would follow that
\[
\begin{aligned}
|U_w(L)-U_w(M)|
&=c_1|A(L)-A(M)|\\
&\le c_1\int_0^1|L(u)-M(u)|\,du\\
&\le c_1\|L-M\|_\infty.
\end{aligned}
\]

For \(0<\varepsilon<1\), let
\[
L_\varepsilon(u):=\frac{\max\{u-\varepsilon,0\}}{1-\varepsilon}.
\]
Its slopes are \(0\) and \(1/(1-\varepsilon)\), and its endpoint
values are zero and one, so \(L_\varepsilon\in\mathcal{L}\).
Subtracting from \(E(u)=u\) gives
\[
E(u)-L_\varepsilon(u)=
\begin{cases}
u,&0\le u\le\varepsilon,\\[3pt]
\dfrac{\varepsilon(1-u)}{1-\varepsilon},
&\varepsilon\le u\le1.
\end{cases}
\]
The first piece increases to \(\varepsilon\), and the second
decreases from \(\varepsilon\). Therefore
\(\|E-L_\varepsilon\|_\infty=\varepsilon\).
Recall that \(C_w=\int_0^1u\,w(u)\,du<\infty\).
Every tail integral is finite, since
\(\int_\varepsilon^1w(u)\,du\le C_w/\varepsilon\).
The corresponding difference quotient satisfies
\[
\begin{aligned}
\frac{U_w(E)-U_w(L_\varepsilon)}{\|E-L_\varepsilon\|_\infty}
&=\frac1\varepsilon\int_0^\varepsilon u\,w(u)\,du
+\frac1{1-\varepsilon}\int_\varepsilon^1(1-u)w(u)\,du\\
&\ge\int_\varepsilon^1(1-u)w(u)\,du\\
&=\int_\varepsilon^1w(u)\,du
-\int_\varepsilon^1u\,w(u)\,du\\
&\ge\int_\varepsilon^1w(u)\,du-C_w
\longrightarrow+\infty
\qquad(\varepsilon\downarrow0).
\end{aligned}
\]

Therefore, it shows every singular weight satisfying the stated conditions gives a continuous
primal-affine representative satisfying
Axioms~\ref{ax:weak-order}--\ref{ax:nontriviality}, but without a
global Lipschitz bound. Thus these nonintegrable weighted orderings cannot be represented by the Lorenz-area functional.

We now examine what happens when dual mixture independence is also imposed.

\subsection{Testing the same ordering in inverse coordinates}

The primal-axiom verification above and the inverse-coordinate
argument below apply more generally to any positive
\(w\in C^1((0,1))\) with \(\int_0^1u\,w(u)\,du<\infty\).

For strictly increasing \(L\), let \(F=L^{-1}\).
Dual mixtures of \(L\)'s become ordinary mixtures of their inverses,
so we will rewrite the same numerical evaluation in terms of \(F\).
At height \(t\), the part of the interval where \(L(u)\ge t\) is
\([F(t),1]\). Its total weight is described by the tail integral
\(H(x):=\int_x^1w(u)\,du\) for \(0<x\le1\).
Integrating these weighted horizontal slices will recover \(U_w(L)\).

For every \(x>0\), \(0\le H(x)\le C_w/x<\infty\).
For integration only, set \(w(0)=w(1)=H(0)=0\). These are values on null sets, not assertions about endpoint limits.
These assignments do not affect the integrals: a strictly increasing \(F\in Y\) takes the value zero only at \(t=0\).
Also \(H(1)=0\), and
\[
0\le H(x)\le\frac1x\int_x^1u\,w(u)\,du\longrightarrow0
\qquad(x\uparrow1).
\]
Thus \(H\) is continuous from the left at \(1\). Since \(w\in C^1((0,1))\), one has \(H\in C^2((0,1)), \allowbreak\; H'(x)=-w(x), \allowbreak\; H''(x)=-w'(x)\).

We now verify the horizontal-slice identity for
\(L\in\mathcal{L}^\uparrow\) and \(F=L^{-1}\):
\[
U_w(L)=\int_0^1 H(F(t))\,dt.
\tag{1}
\]
Indeed, since \(L(u)=\int_0^{L(u)}dt\), Tonelli's theorem gives
\begin{align*}
U_w(L)
&=
\int_0^1 w(u)L(u)\,du\\
&=
\int_0^1\int_0^{L(u)}w(u)\,dt\,du\\
&=
\int_0^1\int_{\{u:\,t\le L(u)\}}w(u)\,du\,dt.
\end{align*}
For \(0<t<1\), because \(L\) is strictly increasing, \(t\le L(u) \iff L^{-1}(t)\le u \iff F(t)\le u\). Therefore
\[
\int_{\{u:\,t\le L(u)\}}w(u)\,du
=
\int_{F(t)}^1w(u)\,du
=
H(F(t)).
\]
Substituting this into the preceding display proves (1).

Thus, on the inverse domain
\[
\begin{aligned}
Y := \Bigl\{ F:[0,1]\to[0,1]: &F \text{ is continuous, strictly increasing, concave},\\
&\ F(0)=0,\ F(1)=1 \Bigr\},
\end{aligned}
\]
the same preorder is represented by
\[
J_H(F):=\int_0^1 H(F(t))\,dt.
\tag{2}
\]

Now suppose the preorder also satisfies the dual package on \(\mathcal{L}^\uparrow\). When we transport the
preorder to \(Y\) by inversion, \(F_1\succsim_Y F_2 \iff F_1^{-1}\succsim F_2^{-1}\), dual independence becomes ordinary mixture independence on the convex set \(Y\) \parencite{yaari1987}. Lemma~\ref{lem:domain-inversion} transports closed-graph continuity to \(Y\), and weak dominance makes \(E\) a greatest element of the transported preorder.
The continuous affine representation result stated in
Subsection~\ref{subsec:affine-setup} therefore supplies a continuous
affine representative \(\mathscr A:Y\to\mathbb R\).
The two functionals \(J_H\) and \(\mathscr A\) represent the same preorder on \(Y\). Therefore,
whenever \(\mathscr A(F)=\mathscr A(G)\), one has \(F\sim G\) under the transported preorder, and hence
\[
J_H(F)=J_H(G).
\tag{3}
\]
This equality will be the constraint we use for the following part, whenever the affine
representative regards two inverse curves as indifferent, the
integral functional \(J_H\) must assign them the same value.

\subsection{Two Local Perturbations}
The above constraint will be applied on a auxiliary family with two parameters.
On that family the affine representative depends on a single affine
coordinate, so fibre constancy will force the integral to have a
Hessian of rank at most one. If \(H''\) is nonzero somewhere, two
perturbations with disjoint supports will instead give two nonzero
diagonal entries and a zero mixed entry, hence rank two.
The following construction makes both calculations legitimate.

A two-direction test will be essential to this argument because curvature along one
direction could be explained by a nonlinear change of numerical scale.
Using two directions with disjoint supports lets us test whether
the two curvatures could arise from dependence on just one affine
coordinate.

Consider a smooth reference function \(F_0\) with
\[
F_0(0)=0,\qquad F_0(1)=1,\qquad
F_0'(t)>0,\qquad F_0''(t)<0\quad(0\le t\le1).
\]
It traverses the whole interval of possible inverse values and has
strictly negative curvature, leaving room for small perturbations
supported away from the endpoints.
A general construction starts from a smooth slope profile
\(v\in C^\infty([0,1])\) with \(v>0\) and \(v'<0\), and sets
\[
F_0(t):=\frac{\int_0^t v(r)\,dr}{\int_0^1 v(r)\,dr}.
\]
The normalization gives \(F_0(0)=0\) and \(F_0(1)=1\), while
\[
F_0'(t)=\frac{v(t)}{\int_0^1v(r)\,dr}>0,
\qquad
F_0''(t)=\frac{v'(t)}{\int_0^1v(r)\,dr}<0.
\]
Every reference curve with the stated properties can be obtained
this way by taking \(v=F_0'\).
Thus \(F_0\in Y\). 
For the sake of argument, suppose that \(H''\not\equiv 0 \allowbreak\;\text{on }(0,1)\). Then there exists \(x_0\in(0,1)\) such that \(H''(x_0)\neq0\). Since \(F_0\) maps \((0,1)\) onto \((0,1)\), choose \(t_0\in(0,1)\) such that \(F_0(t_0)=x_0\). Choose a compact interval \(I=[a,b]\subset(0,1)\) with \(a<t_0<b\), small enough that
\[
H''(F_0(r))H''(F_0(t_0))>0\qquad(r\in I).
\]
Set \(c:=\min_{r\in I}|H''(F_0(r))|>0\).
Then either \(H''(F_0(r))\ge c\) throughout \(I\), or \(H''(F_0(r))\le-c\) throughout \(I\).
\begin{equation*}
|H''(F_0(r))|\ge c\quad(r\in I),\qquad
r\mapsto H''(F_0(r))\text{ has a constant sign on }I.
\tag{4}
\end{equation*}
We place both perturbations inside this same interval \(I\), where
\(H''(F_0(r))\) has a fixed nonzero sign. This prevents cancellation
within either diagonal second derivative.
Their supports will be disjoint so that the mixed derivative vanishes.
Choose any two intervals with
\(a<u_1<v_1<u_2<v_2<b\), and nonzero functions
\(\eta_i\in C_c^\infty((a,b))\) supported in \([u_i,v_i]\).
For an explicit construction, recall the standard smooth function
\(
f(x)=
\begin{cases}
e^{-1/x},&x>0,\\
0,&x\le0.
\end{cases}
\)
The classical limit \(x^{-m}e^{-1/x}\to0\) as \(x\downarrow0\),
for every integer \(m\ge0\), gives
\(f\in C^\infty(\mathbb R)\) with \(f^{(n)}(0)=0\) for all \(n\ge0\).
Now set
\[
\eta_i(r):=f\bigl((r-u_i)(v_i-r)\bigr),
\qquad i=1,2.
\]
Each \(\eta_i\) is smooth as a composition of smooth functions.
Since \((r-u_i)(v_i-r)>0\) exactly when \(u_i<r<v_i\),
we have \(\eta_i(r)>0\) on that interval and \(\eta_i(r)=0\) outside it.
Thus \(\supp\eta_i=[u_i,v_i]\subset(a,b)\), so
\(\eta_i\in C_c^\infty((a,b))\), \(\eta_i\not\equiv0\), and
\[
\supp\eta_1\cap\supp\eta_2=\varnothing.
\tag{5}
\]
The following construction will ensure that both perturbation parameters can vary
independently in an open square while the curve stays in
\(Y\). They control three things: positive slope on \(I\),
negative curvature there, and inverse values bounded away from
zero and one. The last condition will also keep the differentiation
away from any endpoint singularity of \(H\).
Write
\[
B_j(r):=|\eta_1^{(j)}(r)|+|\eta_2^{(j)}(r)|
\qquad(r\in I,\ j=0,1,2).
\]
For each fixed \(r\in I\) and \(\delta>0\),
\[
\max_{|s|,|t|\le\delta}
|s\eta_1^{(j)}(r)+t\eta_2^{(j)}(r)|=\delta B_j(r).
\]
The four pointwise margins therefore determine the parameter threshold
\[
\begin{aligned}
\delta_{\rm adm}^{-1}
:=\max_{r\in I}\max\left\{
\frac{B_1(r)}{F_0'(r)},\frac{B_2(r)}{-F_0''(r)},
\frac{B_0(r)}{F_0(r)},\frac{B_0(r)}{1-F_0(r)}
\right\}\in(0,\infty).
\end{aligned}
\]
Fix any \(0<\delta<\delta_{\rm adm}\), and define
\(F_{s,t}=F_0+s\eta_1+t\eta_2\) for \(|s|,|t|<\delta\).
For \(r\in I\),
\[
\begin{aligned}
F_{s,t}'(r)
&\ge F_0'(r)-\delta B_1(r)
\ge(1-\delta/\delta_{\rm adm})F_0'(r)>0,\\
-F_{s,t}''(r)
&\ge-F_0''(r)-\delta B_2(r)
\ge(1-\delta/\delta_{\rm adm})(-F_0''(r))>0,\\
F_{s,t}(r)
&\ge F_0(r)-\delta B_0(r)
\ge(1-\delta/\delta_{\rm adm})F_0(r)>0,\\
1-F_{s,t}(r)
&\ge1-F_0(r)-\delta B_0(r)
\ge(1-\delta/\delta_{\rm adm})(1-F_0(r))>0.
\end{aligned}
\]
In particular, the common bounds are now established:
\[
\begin{aligned}
R_-&:=\min_{r\in I}\bigl(F_0(r)-\delta B_0(r)\bigr)>0,\\
R_+&:=\max_{r\in I}\bigl(F_0(r)+\delta B_0(r)\bigr)<1,\\
F_{s,t}(I)&\subset[R_-,R_+]\subset(0,1).
\end{aligned}
\]
Outside \(I\), \(F_{s,t}=F_0\).
In particular, \(F_{s,t}'(r)>0\) for \(r<1\), \(F_{s,t}''<0\), and its endpoint values are zero and one.
For \(0\le x<y\le1\), integrating its positive derivative over \((x,y)\) gives \(F_{s,t}(y)>F_{s,t}(x)\).
Therefore \(F_{s,t}\in Y\).

Define
\[
G(s,t):=J_H(F_{s,t})=\int_0^1H(F_{s,t}(r))\,dr.
\tag{6}
\]
The integral is finite by (1), and
\[
G(s,t)=G(0,0)+\int_I\bigl(H(F_{s,t}(r))-H(F_0(r))\bigr)\,dr.
\]
The part of the integral outside \(I\) does not change with \(s,t\).
All derivatives therefore come from the displayed integral over \(I\),
where the arguments of \(H\) remain in a fixed compact subinterval
of \((0,1)\). We can consequently differentiate with respect to the
two parameters since the following estimates justify that differentiation.
Let \(K_j=\max_{R_-\le x\le R_+}|H^{(j)}(x)|<\infty\), for \(j=1,2\).
With \(p=(s,t)\) and \(e_i\) denoting a coordinate direction, any nonzero increment \(h\) whose segment remains in \((-\delta,\delta)^2\) satisfies
\[
\begin{aligned}
\frac{H(F_{p+he_i}(r))-H(F_p(r))}{h}
&=\eta_i(r)\int_0^1H'(F_p(r)+\tau h\eta_i(r))\,d\tau,\\
\frac{H'(F_{p+he_j}(r))\eta_i(r)-H'(F_p(r))\eta_i(r)}{h}
&=\eta_i(r)\eta_j(r)\int_0^1H''(F_p(r)+\tau h\eta_j(r))\,d\tau.
\end{aligned}
\]
Their absolute values are bounded respectively by the integrable functions
\(K_1|\eta_i|\) and \(K_2|\eta_i\eta_j|\) on \(I\).
Dominated convergence gives the first and second parameter derivatives under the integral \parencite{folland1999}.
The same bounds and continuity of \(H'\), \(H''\) give continuity of those derivatives.
Hence \(G\in C^2((-\delta,\delta)^2)\).

\subsection{The incompatible second-order constraints}

We will evaluate what the shared indifference classes require.
Since \(\delta>0\), the parameter square \((-\delta,\delta)^2\)
is a nonempty, open, convex subset of \(\mathbb R^2\).
Recall that \(F_{s,t}=F_0+s\eta_1+t\eta_2\).
For \((s,t),(s',t')\) in this square and \(\lambda\in[0,1]\),
their convex combination remains in the square, and
\[
F_{\lambda s+(1-\lambda)s',\,\lambda t+(1-\lambda)t'}
=\lambda F_{s,t}+(1-\lambda)F_{s',t'}.
\]
All these curves belong to \(Y\), as verified above. Affinity
of \(\mathscr A\) on \(Y\) therefore gives
\[
\mathscr A\!\left(F_{\lambda s+(1-\lambda)s',\,\lambda t+(1-\lambda)t'}\right)
=\lambda\mathscr A(F_{s,t})+(1-\lambda)\mathscr A(F_{s',t'}).
\]
Thus \((s,t)\mapsto\mathscr A(F_{s,t})\) is a real-valued affine
function on that open convex square. Therefore we could apply Lemma~\ref{lem:finite-affine} with \(m=2\), so it has the form
\[
\mathscr A(F_{s,t})=a_0+a_1s+a_2t
\tag{7}
\]
for some constants \(a_0,a_1,a_2\in\mathbb R\). By (3), equality of the affine expression \(a_0+a_1s+a_2t\) forces equality of \(G(s,t)\). Thus \(G\) is constant on the affine level sets of
\[
(s,t)\longmapsto a_0+a_1s+a_2t.
\tag{8}
\] 

If \(a_1=a_2=0\), then \(\mathscr A(F_{s,t})\) is constant on the parameter square \((-\delta,\delta)^2\). By (3), \(G\) is also
constant on that square, so the Hessian of \(G\) is zero. In particular,
\[
\det D^2G(0,0)=0.
\tag{9}
\]
If \((a_1,a_2)\ne(0,0)\), first suppose \(a_1\ne0\).
On the line \(t=0\), the equation \(a_0+a_1s=x\) gives
\(s=(x-a_0)/a_1\). Accordingly, set
\[
\delta_1:=\frac{|a_1|\delta}{|a_1|+|a_2|}\in(0,\delta],\qquad
I_1:=(a_0-|a_1|\delta,a_0+|a_1|\delta),\qquad
\varphi(x):=G\!\left(\frac{x-a_0}{a_1},0\right)\quad(x\in I_1).
\]
For \(x\in I_1\), the defining point lies in the parameter square
and depends affinely on \(x\). Since \(G\in C^2\), this proves
\(\varphi\in C^2(I_1)\).
The smaller radius \(\delta_1\) lets us compare a nearby point
\((s,t)\) with a point on \(t=0\) inside the same square:
if \(|s|,|t|<\delta_1\),
\[
\left|s+\frac{a_2}{a_1}t\right|
\le |s|+\frac{|a_2|}{|a_1|}|t|
<\delta_1\frac{|a_1|+|a_2|}{|a_1|}=\delta.
\]
Thus both \((s,t)\) and \((s+(a_2/a_1)t,0)\) lie in the original parameter square, and they have the same value of \(a_0+a_1s+a_2t\).
By (3), these equal affine coordinates give equal values of \(G\), so
\[
G(s,t)=G\!\left(s+\frac{a_2}{a_1}t,0\right)=\varphi(a_0+a_1s+a_2t)
\qquad(|s|,|t|<\delta_1).
\]
If \(a_1=0\), then \(a_2\ne0\).
In this case define \(\varphi(x)=G(0,(x-a_0)/a_2)\) for \(|x-a_0|<|a_2|\delta\);
\((s,t)\) and \((0,t)\) are both in the original parameter square
and have the same affine coordinate. Thus they give the same value
of \(G\), and the same factorization holds. Here \(\varphi\) is \(C^2\)
because it is defined from \(G\) along the line \(s=0\).
In either nonconstant case, the local identity
\(G(s,t)=\varphi(a_0+a_1s+a_2t)\) gives, by the chain rule,
\[
\begin{aligned}
\frac{\partial G}{\partial s}(s,t)
&=a_1\varphi'(a_0+a_1s+a_2t),\\
\frac{\partial G}{\partial t}(s,t)
&=a_2\varphi'(a_0+a_1s+a_2t).
\end{aligned}
\]
These formulas hold on the neighborhood where the factorization
was established. Since \(a_1,a_2\) are constants, differentiating
once more and evaluating at \((0,0)\) gives
\[
\begin{aligned}
\frac{\partial^2G}{\partial s^2}(0,0)&=a_1^2\varphi''(a_0),\\
\frac{\partial^2G}{\partial s\partial t}(0,0)&=a_1a_2\varphi''(a_0),\\
\frac{\partial^2G}{\partial t^2}(0,0)&=a_2^2\varphi''(a_0).
\end{aligned}
\]
Consequently,
\[
D^2G(0,0)=\varphi''(a_0)
\begin{pmatrix}a_1^2&a_1a_2\\a_1a_2&a_2^2\end{pmatrix},
\]
and its determinant is
\[
\begin{aligned}
\det D^2G(0,0)
&=(a_1^2\varphi''(a_0))(a_2^2\varphi''(a_0))
-(a_1a_2\varphi''(a_0))^2\\
&=\varphi''(a_0)^2(a_1^2a_2^2-a_1^2a_2^2)=0.
\end{aligned}
\]
Together with the constant case, this proves
\[
\det D^2G(0,0)=0.
\tag{10}
\]

Now calculate the same Hessian terms directly from the integral in (6). 
The diagonal entries represent the curvature produced by each bump
separately; the off diagonal mixed entry represents their interaction.
The fixed sign of \(H''(F_0(r))\) will make each diagonal entry
nonzero, while disjoint supports eliminate the interaction.

Both perturbations vanish outside \(I\), and \(F_{s,t}(I)\) stays
in \([R_-,R_+]\subset(0,1)\). Thus all parameter dependence occurs
where \(H'\) and \(H''\) are bounded, and the estimates already
established justify differentiation under the integral.
Recall that \(F_{s,t}=F_0+s\eta_1+t\eta_2\), so the chain rule gives
\[
\begin{aligned}
\frac{\partial G}{\partial s}(s,t)
&=\int_I H'(F_{s,t}(r))\,\eta_1(r)\,dr,\\
\frac{\partial G}{\partial t}(s,t)
&=\int_I H'(F_{s,t}(r))\,\eta_2(r)\,dr.
\end{aligned}
\]
The functions \(\eta_1,\eta_2\) do not depend on \(s,t\).
A second differentiation therefore introduces one more perturbation
factor. Evaluating at \((s,t)=(0,0)\), where \(F_{s,t}=F_0\), gives
\begin{align*}
\frac{\partial^2G}{\partial s^2}(0,0)
&=\int_I H''(F_0(r))\,\eta_1(r)^2\,dr, \tag{11}\\
\frac{\partial^2G}{\partial s\partial t}(0,0)
&=\int_I H''(F_0(r))\,\eta_1(r)\eta_2(r)\,dr, \tag{12}\\
\frac{\partial^2G}{\partial t^2}(0,0)
&=\int_I H''(F_0(r))\,\eta_2(r)^2\,dr. \tag{13}
\end{align*}

For the mixed entry, the two test functions have disjoint supports,
so \(\eta_1(r)\eta_2(r)=0\) for every \(r\) by (5).
The integral in (12) is therefore identically zero:
\[
\frac{\partial^2G}{\partial s\partial t}(0,0)=0.
\tag{14}
\]

For the diagonal entries, recall from (4) that \(H''(F_0(r))\)
has a fixed nonzero sign on \(I\), with \(|H''(F_0(r))|\ge c>0\).
There is consequently no cancellation in either integral:
for \(i=1,2\),
\[
\begin{aligned}
\left|\int_I H''(F_0(r))\,\eta_i(r)^2\,dr\right|
&=\int_I |H''(F_0(r))|\,\eta_i(r)^2\,dr\\
&\ge c\int_I\eta_i(r)^2\,dr>0.
\end{aligned}
\]
The last inequality holds because each \(\eta_i\) is continuous
and nonzero, so its square has a strictly positive integral.
Thus both diagonal entries have the same sign as \(H''(F_0)\), and
\[
\frac{\partial^2G}{\partial s^2}(0,0)\neq0,
\qquad
\frac{\partial^2G}{\partial t^2}(0,0)\neq0.
\tag{15}
\]

Together, (14) and (15) make the Hessian diagonal with two nonzero
entries of the same sign:
\[
D^2G(0,0)=
\begin{pmatrix}
\displaystyle \int_IH''(F_0(r))\eta_1(r)^2\,dr & 0\\[12pt]
0 & \displaystyle \int_IH''(F_0(r))\eta_2(r)^2\,dr
\end{pmatrix}.
\]
Its determinant is therefore
\[
\det D^2G(0,0)
=
\frac{\partial^2G}{\partial s^2}(0,0)
\frac{\partial^2G}{\partial t^2}(0,0)>0,
\]
contradicting (10), which requires this determinant to be zero.
The contradiction concerns two properties of the very same function
\(G\): its indifference classes force a singular Hessian, whereas
its integral formula gives an invertible Hessian.
Thus no affine representative can share the ordering of \(J_H\)
when \(H''\) is nonzero somewhere.

\subsection{Singular weights and the Lorenz-area representation}
Therefore, the additional constraint imposed by the axiom of dual independence forces
\(H''\equiv0\) on \((0,1)\). Now let's translate this constraint back
to the weight \(w\).

First, for \(0<x<y<1\), the fundamental theorem of calculus gives
\[
H'(y)-H'(x)=\int_x^yH''(t)\,dt=0.
\]
Thus \(H'\) is a constant \(b\). Integrating once more gives
\(H(x)=a+bx\) on \((0,1)\), for some \(a\in\mathbb R\).

Next, recall that \(H(1)=0\) and that left continuity at \(1\)
has already been proved. Hence
\[
a+b=\lim_{x\uparrow1}(a+bx)
=\lim_{x\uparrow1}H(x)=H(1)=0.
\]
Writing \(c:=-b=a\), we obtain
\[
H(x)=c(1-x),
\qquad
w(x)=-H'(x)=c
\qquad(0<x<1).
\]
The weight is therefore constant, and its assumed positivity gives
\(c>0\). Consequently, for every \(L\in\mathcal{L}\),
\[
U_w(L)=\int_0^1 w(u)L(u)\,du
=c\int_0^1 L(u)\,du.
\]
Within the class of weights considered here, dual independence
therefore permits only positive constant weights, which give the
Lorenz-area order. Every nonconstant candidate consequently fails
dual independence, including a singular weight that admits a finite
continuous monotone kernel. The uniform Lipschitz bound also follows
for the surviving constant weights:
\[
|U_w(L)-U_w(M)|
=c|A(L)-A(M)|\le c\|L-M\|_\infty
\qquad(L,M\in\mathcal{L}).
\]
Which contradicts the construction of the singular kernel weight as it cannot satisfy dual independence.

For every member of the singular family \(w_p(u)=u^{-p}\), the earlier calculations verified all the primal
requirements. Its tail integral \(H_p(x):=\int_x^1u^{-p}\,du\), however,
satisfies
\[
H_p'(x)=-x^{-p},
\qquad
H_p''(x)=p\,x^{-p-1}>0
\qquad(0<x<1).
\]
This is precisely the nonzero curvature ruled out above, so none
of these singular weights can satisfy dual independence.

The above construction is intended to provide a self-explanatory account of why the original singular weighted-integral route fails to produce a counterexample to the combined theorem. To disprove the representation theorem, it is natural to construct a weighted-integral example that induces singular behavior at the endpoint which sabotaged the area integral representation. Unexpectedly, however, once the dual mixture structure is imposed simultaneously, the singular weighted-integral ansatz comes under a rigidity condition and loses its non-affine freedom. The combined effects of primal and dual independence collapse every admissible weight back to the form of integral Lorenz-area ordering representation. Of course, the sheer failure of one natural possible route to disprove a theorem does not itself constitute an affirmative proof of its truth. On the other hand, these calculations motivate the idea that the combined mixture structure indeed forces a more rigid constraint on the final representation.

In the remaining paragraphs of section 2, we will provide a detailed description of the basic idea underlying the main proof project  before diving into actual mathematical arguments. The main proof starts with applying the classical affine-representation theorem in the two mixture structures, giving continuous affine representatives. Note that the representation here is a general affine functional rather than a more restricted integral form. However, it's easy to notice that the Lorenz domain is not very well behaved, as it is closed under neither addition nor arbitrary scalar multiplication, therefore it's hard to make any progress by the usual analytic or algebraic arguments. As a result, we need to move temporarily to another domain in order to build the desired auxiliary tool that will induce the rigid characterization. More specifically, the proof will pass from affine representations on curve domains to linear functionals on a space of smooth endpoint-vanishing directions, using parametrized smooth perturbations.

We will work on the direction space \(V\), which consists of smooth functions that vanish at both endpoints. Such a choice is motivated by the property of the perfect equality curve \(E\). So for an arbitrary Lorenz curve, the constructed gap \(E-L\)  will have vanishing endpoint values, which satisfies the requirement for being a member of direction space \(V\). Meanwhile, this space is closed under addition and scalar multiplication.  A reference gap \(q\in V\) will be constructed so the objects \(E-sq\) and \(E+sq\) have positive slopes and
strictly signed curvatures in their respective domains, leaving
room for perturbations in any fixed smooth direction. With the help of a perturbation argument, we can apply identities
derived on the space of smooth directions to admissible Lorenz
curves, even when their derivation involves forbidden
operations which the Lorenz domain is not closed.

Affine property along admissible segments will give rise to linear functionals $T$  and $m$ on the direction space. They recover the gap functionals in the two representations on both the primal and dual sides. We choose two admissible comparison curves with the same primal evaluation. Since the two representatives describe the same ordering, their dual evaluations must also agree. Comparing the inverse expansions first shows that the two functionals agree on the comparison directions and then gives a quadratic identity. Applying perturbation argument and functional-analytic techniques allows us to pass from the quadratic identity to a bilinear identity for the linear functional $T$ .

Subsequently, we will established the classification result that the final representation of linear functional $T$  can only take one of two forms: an area integral representation or one of two endpoint derivative representations. Given that all representations in the form of endpoint derivatives contradict the continuity of the original representing function $U$ at $E$, the endpoint cases are excluded, the form of T can only be the normalized area integral. Using the relationship formula between $T$ and $U$, we then complete the proof of the final representation theorem.

\section{Proof of the main theorem: affine representations and the bilinear identity}
\label{sec:proof-main}
\setcounter{equation}{15}

\subsection{Affine-representation setup}
\label{subsec:affine-setup}

Fix an arbitrary relation \(\succsim\) satisfying
Axioms~\ref{ax:weak-order}--\ref{ax:dual-independence}.
We first obtain continuous affine representatives on the primal and inverse
domains. These representatives remain fixed throughout the proof,
including when we later restrict attention to smooth curves.

We use the continuous affine-representation result for Lorenz curves proved in
\textit{Measure of Morality}
\parencite[Proposition~1]{wei2026measuremoralitymathematicaltheory}.
Its classical antecedent is the mixture-space representation theorem
of \textcite[Theorem~8]{hersteinMilnor1953}; see also
\textcite[Section~8.4]{fishburn1970}.
On \(\mathcal{L}\), it gives a nonconstant affine representative,
continuous in the relative uniform topology, under completeness,
transitivity, mixture independence, nontriviality, and closed-graph
continuity. Its proof uses only convexity of the underlying curve
domain, so the same construction applies to the convex inverse
domain \(Y\).
Note that this result is about functional representation rather than integral representation.
The one-sided integral-representation construction in the previous paper
uses an additional quantitative Archimedean condition.

By Lemma~\ref{lem:domain-inversion} in the appendix, \(\mathcal{L}\) and
\(Y\) are convex, and inversion
\(\iota:\mathcal{L}^\uparrow\to Y\), \(\iota(L)=L^{-1}\),
is a homeomorphism for the relative uniform topologies.
Transport the restriction of the ordering to \(Y\) by
\[
F\succsim_Y G
\quad\Longleftrightarrow\quad F^{-1}\succsim G^{-1}
\qquad(F,G\in Y).
\]
Completeness and transitivity pass to this relation, and its graph
is closed by the inversion homeomorphism.
Nontriviality on \(Y\) follows from
Axiom~\ref{ax:nontriviality}.
For \(F,H\in Y\) and \(\lambda\in[0,1]\),
\[
(\lambda F+(1-\lambda)H)^{-1}
=F^{-1}\oplus_\lambda H^{-1}.
\]
Thus dual independence becomes ordinary mixture independence
on \(Y\). Applying that construction on the two domains
therefore supplies nonconstant continuous affine functions
\(U:\mathcal{L}\to\mathbb R\) and
\(W_Y:Y\to\mathbb R\) satisfying
\[
\begin{aligned}
L\succsim M&\Longleftrightarrow U(L)\ge U(M)
&&(L,M\in\mathcal{L}),\\
F\succsim_Y G&\Longleftrightarrow
W_Y(F)\ge W_Y(G)
&&(F,G\in Y).
\end{aligned}
\]

For later use, note that continuity in the uniform topology means that for either
\((X,V)=(\mathcal{L},U)\) or \((Y,W_Y)\),
every \(x\in X\), and every
\(\varepsilon>0\), there exists \(\delta>0\) such that
\begin{equation}
y\in X,\quad \|y-x\|_\infty<\delta
\quad\Longrightarrow\quad |V(y)-V(x)|<\varepsilon.
\label{eq:representative-continuity}
\end{equation}

\smallskip
\noindent\emph{Comparison on the common domain.}
Define
\[
\begin{aligned}
W(L)&:=W_Y(L^{-1})&&(L\in\mathcal{L}^\uparrow),\\
D_U(L)&:=U(E)-U(L)&&(L\in\mathcal{L}),\\
D_W(L)&:=W(E)-W(L)&&(L\in\mathcal{L}^\uparrow).
\end{aligned}
\]
Continuity of inversion makes \(W\) continuous on
\(\mathcal{L}^\uparrow\), and affinity of \(W_Y\) gives
\[
W(L\oplus_\lambda M)=\lambda W(L)+(1-\lambda)W(M)
\qquad(L,M\in\mathcal{L}^\uparrow,\ \lambda\in[0,1]).
\]
Since every Lorenz curve lies below \(E\), weak dominance gives
\(D_U\ge0\) and \(D_W\ge0\) on their respective domains.
The two representatives encode the same ordering on
\(\mathcal{L}^\uparrow\). For \(L,M\in\mathcal{L}^\uparrow\), this gives
\begin{equation}
\begin{aligned}
D_U(L)=D_U(M)&\Longleftrightarrow L\sim M
\Longleftrightarrow D_W(L)=D_W(M),\\
D_U(L)<D_U(M)&\Longleftrightarrow L\succ M
\Longleftrightarrow D_W(L)<D_W(M).
\end{aligned}
\label{eq:common-deficiency-order}
\end{equation}
Strict increase is needed for the ordinary inverse to exist;
it does not imply smoothness. The inverse calculations below use
smooth test curves with strictly positive derivatives.
The final approximation argument concerns \(U\) on \(\mathcal{L}\),
so its conclusion also covers curves with horizontal segments.

\subsection{Smooth directions and shape-preserving approximation}
\label{subsec:smooth-directions}

The next two lemmas prepare the test directions and approximation
tools used in the core arguments. Smoothness is required properties of these auxiliary
functions, shape-preserving density will let us recover every curve
in the original domain at the end.

All derivatives at \(0,1\) are one-sided. Let
\[
\begin{aligned}
V&:=\{h\in C^\infty([0,1]):h(0)=h(1)=0\},\\
A(f)&:=\int_0^1f(x)\,dx&&(f\in C([0,1])),\\
\|f\|_{C^j}&:=\max_{0\le i\le j}\|f^{(i)}\|_\infty
&&(f\in C^j([0,1])).
\end{aligned}
\]
Here \(C^\infty([0,1])\) consists of the functions whose derivatives
of every order on \((0,1)\) extend continuously to \([0,1]\).

For the rest of paper, fix an arbitrary reference auxiliary function \(q\in V\) with
\(q''(x)<0\) for all \(x\in[0,1]\). Let's rescale it so that \(\|q'\|_\infty=1\); this is possible because
\(q'\) is not identically zero. The continuous positive function
\(-q''\) has a positive minimum and a finite maximum on \([0,1]\).
Set
\[
a_q:=\min_{[0,1]}(-q'')>0,\qquad
b_q:=\max_{[0,1]}(-q'')<\infty,\qquad
C_q:=\frac1{2q(1/2)}.
\]
The normalization fixes only the scale, the shape of \(q\) remains
arbitrary so the lower curvature bound could be constructed in a way that for the purpose of  perturbations argument that would be employed in the rest of this paper. 

For completeness, \(w(x):=q(x)-(a_q/2)x(1-x)\) satisfies
\(w''=q''+a_q\le0\) and \(w(0)=w(1)=0\). Hence \(w\) is concave
and nonnegative; the elementary curvature criterion is verified
in \eqref{eq:curvature-shape} in Appendix~\ref{app:bernstein-density}. Indeed, a concave function
lies above the chord joining its endpoints, which here is zero.
Therefore
\[
q(x)\ge\frac{a_q}{2}x(1-x)>0\qquad(0<x<1).
\]
In particular, \(q(1/2)>0\), so \(C_q\) is finite and positive.
This lower bound also controls the endpoint slopes. Divide it by
\(x\) near zero and by \(1-x\) near one, then take the corresponding
one-sided limits:
\[
\begin{aligned}
q'(0)&=\lim_{x\downarrow0}\frac{q(x)}x\ge\frac{a_q}{2}>0,\\
-q'(1)&=\lim_{x\uparrow1}\frac{q(x)}{1-x}\ge\frac{a_q}{2}>0.
\end{aligned}
\]
Integrating the same lower bound over the whole interval gives
\[
A(q)\ge\frac{a_q}{2}\int_0^1x(1-x)\,dx=\frac{a_q}{12}>0.
\]
In particular, the two endpoint zeros are simple. These properties will be used in the classification.

Several constructions below divide by endpoint factors or express
inversion errors as parameter integrals. Lemma~\ref{lem:parameter-integration}
in Appendix~\ref{app:parameter-integration} justifies differentiating
those integrals while retaining regularity at the endpoints.

The endpoint zeros of gap function \(q\) are set in such a way that we will have a common  scale for all directions in \(V\). Indeed each such direction is bounded  by
a finite multiple of \(q\). The next lemma proves this comparison
and identifies the endpoint values of the quotient.
Positivity will later turn it into a bound on the linear functionals.

\begin{lemma}[Division by the reference gap]
\label{lem:quadratic-division}
For \(h\in V\), \(h/q\) has a unique extension in
\(C^\infty([0,1])\), and
\begin{equation}
(h/q)(0)=\frac{h'(0)}{q'(0)},\qquad
(h/q)(1)=\frac{h'(1)}{q'(1)},\qquad
\|h/q\|_\infty\le C_q\|h'\|_\infty.
\label{eq:quadratic-division-bound}
\end{equation}
The constant \(C_q\) is optimal.
\end{lemma}
\begin{proof}
For \(f=h\) or \(f=q\), define
\[
A_f(x):=\int_0^1f'(rx)\,dr,\qquad
B_f(x):=-\int_0^1f'(x+r(1-x))\,dr.
\]
The fundamental theorem of calculus and the endpoint values give
\[
f(x)=xA_f(x)=(1-x)B_f(x).
\]
These identities factor out the zero at the left and right
endpoint, respectively. The remaining factors stay smooth:
Lemma~\ref{lem:parameter-integration} shows that \(A_f,B_f\)
are smooth up to both endpoints. Explicitly,
\[
\begin{aligned}
A_f^{(j)}(x)&=\int_0^1r^jf^{(j+1)}(rx)\,dr,\\
B_f^{(j)}(x)&=-\int_0^1(1-r)^j
f^{(j+1)}(x+r(1-x))\,dr
\qquad(j\ge0).
\end{aligned}
\]
Since \(q>0\) on \((0,1)\), \(q'(0)>0\), and \(q'(1)<0\),
\[
\begin{aligned}
A_q(x)&=q(x)/x>0&&(0<x<1),&
A_q(0)&=q'(0)>0,\\
B_q(x)&=q(x)/(1-x)>0&&(0<x<1),&
B_q(1)&=-q'(1)>0.
\end{aligned}
\]
Consequently \(A_h/A_q\) on \([0,1)\) and \(B_h/B_q\) on
\((0,1]\) are smooth and agree with \(h/q\) in the interior.
The first quotient extends through zero and the second through one;
their denominators are positive at the respective endpoints.
Since the two formulas agree on their overlap, they define the
required extension, with
\[
\frac{A_h(0)}{A_q(0)}=\frac{h'(0)}{q'(0)},\qquad
\frac{B_h(1)}{B_q(1)}=\frac{h'(1)}{q'(1)}.
\]
Uniqueness follows from continuity and agreement on \((0,1)\).

Finally, the same endpoint integral formulas give two bounds
for \(|h(x)|\). Integration from zero gives \(x\|h'\|_\infty\),
while integration from one gives \((1-x)\|h'\|_\infty\).
Taking the smaller of these bounds yields
\[
|h(x)|\le\min\{x,1-x\}\|h'\|_\infty.
\]
Concavity compares \(q\) with the chords joining its midpoint to
the endpoints. On the left half of the interval the chord has
height \(2xq(1/2)\); on the right half it has height
\(2(1-x)q(1/2)\). Therefore
\[
q(x)\ge2q(1/2)\min\{x,1-x\}\qquad(0\le x\le1).
\]
Dividing by \(q(x)>0\) in the interior gives an upper bound
for the quotient below, and equality holds at the midpoint:
\[
\max_{0<x<1}\frac{\min\{x,1-x\}}{q(x)}
=\frac1{2q(1/2)}=C_q.
\]
Thus, for \(0<x<1\),
\[
\frac{|h(x)|}{q(x)}
\le\frac{\min\{x,1-x\}}{q(x)}\|h'\|_\infty
\le C_q\|h'\|_\infty.
\]
The endpoint inequalities follow by continuity.

To verify optimality within the smooth space \(V\), for
\(\varepsilon>0\) set
\[
h_\varepsilon(x):=
\sqrt{\tfrac14+\varepsilon^2}
-\sqrt{(x-\tfrac12)^2+\varepsilon^2}.
\]
For each \(\varepsilon>0\), this function is smooth and vanishes
at both endpoints. Differentiating it gives
\[
h_\varepsilon\in V,\qquad
\|h_\varepsilon'\|_\infty
=\frac{1/2}{\sqrt{1/4+\varepsilon^2}}.
\]
Apply the proved bound to \(h_\varepsilon\), and bound the
supremum from below by its value at the midpoint:
\[
C_q\ge\frac{\|h_\varepsilon/q\|_\infty}
{\|h_\varepsilon'\|_\infty}
\ge\frac{h_\varepsilon(1/2)}
{q(1/2)\|h_\varepsilon'\|_\infty}
\xrightarrow[\varepsilon\downarrow0]{}\frac1{2q(1/2)}=C_q.
\]
Thus any smaller constant would fail for sufficiently small
\(\varepsilon\), proving optimality within \(V\).
\end{proof}

\subsection{The positive linear parts of the two representatives}
\label{subsec:positive-linear-parts}

Affinity on a convex curve domain does not yet give a linear map on
all smooth endpoint-vanishing directions. The next proposition
constructs that map from admissible segments. Dominance then gives
positivity, and comparison with \(q\) supplies the estimate needed
to control the inversion remainder.

\begin{proposition}[Directional linearity and normalization]
\label{prop:positive-directions}
There exist unique linear maps \(T_0,m_0:V\to\mathbb R\)
such that
\begin{align}
U(L)-U(E)&=T_0(L-E)
&&(L\in\mathcal{L}\cap C^\infty([0,1])),
\label{eq:original-primal-direction}\\
W_Y(E)-W_Y(F)&=m_0(F-E)
&&(F\in Y\cap C^\infty([0,1])).
\label{eq:original-dual-direction}
\end{align}
Moreover,
\[
h\in V,\quad h\ge0
\quad\Longrightarrow\quad T_0(h)\ge0,\quad m_0(h)\ge0.
\]
The values at \(q\) give strictly positive scales:
\[
\sigma_U:=T_0(q)>0,\qquad \sigma_W:=m_0(q)>0.
\]
Define the normalized functionals by
\[
T:=\frac{1}{\sigma_U}T_0,\qquad
m:=\frac{1}{\sigma_W}m_0.
\]
Then \(T(q)=m(q)=1\)
and
\begin{equation}
\max\{|T(h)|,|m(h)|\}
\le\|h/q\|_\infty\le C_q\|h'\|_\infty
\qquad(h\in V).
\label{eq:positive-weighted-bound}
\end{equation}
\end{proposition}
\begin{proof}
\noindent\emph{(i) Two positive affine maps.}
Inversion reverses pointwise order. Taking \(-W_Y\) in the
inverse coordinates puts both representatives in the same convention:
a pointwise increase produces a nonnegative change in value.

Set
\[
\begin{array}{c|cc}
\nu&X_\nu&V_\nu\\ \hline
-1&\mathcal{L}\cap C^\infty([0,1])&U|_{X_{-1}}\\
+1&Y\cap C^\infty([0,1])&-W_Y|_{X_{+1}}.
\end{array}
\]
For \(F,G\in Y\), Lemma~\ref{lem:domain-inversion} gives
\(F^{-1},G^{-1}\in\mathcal{L}^\uparrow\), so the preference
comparison below is on the domain represented by \(W\).
For \(F\ge G\) and \(y\in[0,1]\),
\[
\begin{aligned}
x:=F^{-1}(y)&\ \Longrightarrow\ G(x)\le F(x)=y\\
&\ \Longrightarrow\ x\le G^{-1}(y),\\
F^{-1}\le G^{-1}
&\ \Longrightarrow\ G^{-1}\succsim F^{-1}\\
&\ \Longrightarrow\ W_Y(G)\ge W_Y(F).
\end{aligned}
\]
Together with primal dominance, this proves, for every
\(\nu\in\{-1,+1\}\) and \(P,Q\in X_\nu\),
\begin{equation}
P\ge Q\quad\Longrightarrow\quad V_\nu(P)\ge V_\nu(Q).
\label{eq:two-chart-positivity}
\end{equation}
Both \(V_\nu\) are affine on their convex domains.

\noindent\emph{(ii) Admissible signed directions.}
Given the nature of equality curve, it has no built-in curvature margin, so arbitrary signed perturbations of Lorenz curve that corresponds to perfect equality will not be automatically admissible under our machinery. Therefore we need to control all the analytical estimates manually to ensure that curves \(B_\nu\)
chosen below fall within the corresponding admissible margins when it come to both slope and curvature.

Indeed, let's fix \(\tau\in(0,1)\). Since \(\|q'\|_\infty=1\) and
\(a_q=\min_{[0,1]}(-q'')\), the common margins of the two
reference curves are
\[
\begin{aligned}
B_\nu:=E+\nu\tau q,\qquad
B_\nu'&=1+\nu\tau q',\qquad
-\nu B_\nu''=-\tau q'',\\
\min_{\nu\in\{-1,+1\}}\min_{x\in[0,1]}B_\nu'(x)
&=1-\tau>0,\\
\min_{x\in[0,1]}(-\nu B_\nu''(x))&=\tau a_q>0.
\end{aligned}
\]
For an arbitrary fixed \(h\in V\), define the symmetric
interval allowed by these uniform bounds:
\[
I_\tau(h):=
\{t\in\mathbb R:|t|\|h'\|_\infty<1-\tau,
\ |t|\|h''\|_\infty<\tau a_q\}.
\]
With the convention \(a/0:=+\infty\) for \(a>0\), define its radius by
\[
r_\tau(h):=
\min\left\{\frac{1-\tau}{\|h'\|_\infty},
\frac{\tau a_q}{\|h''\|_\infty}\right\}\in(0,\infty].
\]
Since \(h\) is smooth, both derivative norms are finite, so
\(r_\tau(h)>0\); for \(h=0\), the convention gives
\(r_\tau(0)=+\infty\). The radius is chosen for each direction,
and thus every \(h\in V\) has admissible nonzero steps of both signs.
Then \(I_\tau(h)=(-r_\tau(h),r_\tau(h))\), and
\[
\begin{aligned}
(B_\nu+th)'&\ge1-\tau-|t|\|h'\|_\infty>0,\\
-\nu(B_\nu+th)''&\ge\tau a_q-|t|\|h''\|_\infty>0,\\
(B_\nu+th)(0)&=0,\qquad (B_\nu+th)(1)=1
\qquad(t\in I_\tau(h)).
\end{aligned}
\]
Equation~\eqref{eq:curvature-shape} therefore gives
\[
B_{-1}+th\in\mathcal{L}^\uparrow\cap C^\infty([0,1]),\qquad
B_{+1}+th\in Y\cap C^\infty([0,1])
\qquad(t\in I_\tau(h)).
\]

\noindent\emph{(iii) Definition and parameter independence.}
The value change along each admissible segment will be evaluated. As the affine property makes the difference quotient independent of the step
size, it will yield an exact directional slope.

Fix \(h\in V\), keeping \(\nu\) and \(\tau\) fixed as well;
only the step size will vary in this part.
Part~(ii) ensures that \(B_\nu+th\in X_\nu\)
for every \(t\in I_\tau(h)\). For \(t\ne0\) in this interval, define
\[
Q_\nu(h;t):=
\frac{V_\nu(B_\nu+th)-V_\nu(B_\nu)}t.
\]
First compare two positive steps \(0<s<t<r_\tau(h)\).
Since \(s/t\in(0,1)\), the shorter step lies on the segment from
\(B_\nu\) to \(B_\nu+th\), and affinity gives
\[
\begin{aligned}
B_\nu+sh&=(1-s/t)B_\nu+(s/t)(B_\nu+th),\\
V_\nu(B_\nu+sh)
&=(1-s/t)V_\nu(B_\nu)+(s/t)V_\nu(B_\nu+th),\\
V_\nu(B_\nu+sh)-V_\nu(B_\nu)
&=\frac st\bigl(V_\nu(B_\nu+th)-V_\nu(B_\nu)\bigr).
\end{aligned}
\]
Dividing by \(s>0\), we obtain
\[
\begin{aligned}
Q_\nu(h;s)
&=\frac{V_\nu(B_\nu+sh)-V_\nu(B_\nu)}s
&&\text{(definition)}\\
&=\frac1s\frac st\bigl(V_\nu(B_\nu+th)-V_\nu(B_\nu)\bigr)
&&\text{(affinity above)}\\
&=\frac{V_\nu(B_\nu+th)-V_\nu(B_\nu)}t
&&\text{(cancel \(s\ne0\))}\\
&=Q_\nu(h;t)
&&\text{(definition)}.
\end{aligned}
\]
Thus any two positive admissible steps give the same quotient,
interchanging them if necessary.

To compare opposite signs, use the symmetry of \(I_\tau(h)\).
For \(0<t<r_\tau(h)\), both \(B_\nu+th\) and \(B_\nu-th\)
belong to \(X_\nu\), so affinity at their midpoint gives
\[
\begin{aligned}
B_\nu&=\tfrac12(B_\nu+th)+\tfrac12(B_\nu-th),\\
2V_\nu(B_\nu)&=V_\nu(B_\nu+th)+V_\nu(B_\nu-th),\\
V_\nu(B_\nu-th)-V_\nu(B_\nu)
&=-\bigl(V_\nu(B_\nu+th)-V_\nu(B_\nu)\bigr).
\end{aligned}
\]
Here both the value increment and the step change sign. Hence
\[
\begin{aligned}
Q_\nu(h;-t)
&=\frac{V_\nu(B_\nu-th)-V_\nu(B_\nu)}{-t}\\
&=\frac{-\bigl(V_\nu(B_\nu+th)-V_\nu(B_\nu)\bigr)}{-t}
=Q_\nu(h;t).
\end{aligned}
\]
Combining the two comparisons, for arbitrary nonzero
\(s,t\in I_\tau(h)\), the numbers \(|s|,|t|\) are positive
admissible steps, and:
\[
Q_\nu(h;s)=Q_\nu(h;|s|)=Q_\nu(h;|t|)=Q_\nu(h;t)
\]
Define \(\Lambda_\nu(h)\) to be this common value. Equivalently,
\[
V_\nu(B_\nu+th)-V_\nu(B_\nu)=t\Lambda_\nu(h)
\qquad(t\in I_\tau(h)),
\]
where the equality at \(t=0\) holds directly.

\noindent\emph{(iv) Linearity.}
We now compare different directions using this exact increment
formula. For the zero direction, \(B_\nu+t0=B_\nu\), so
\(\Lambda_\nu(0)=0\).

For homogeneity, let \(a\ne0\). Since
\(\|(ah)^{(j)}\|_\infty=|a|\|h^{(j)}\|_\infty\) for \(j=1,2\),
\[
r_\tau(ah)=\frac{r_\tau(h)}{|a|},\qquad
t\in I_\tau(ah)\ \Longleftrightarrow\ at\in I_\tau(h).
\]
Choose \(0<|t|<r_\tau(h)/|a|\). Both \(t\) and \(at\) are
nonzero admissible steps for their respective directions. Thus
\[
\begin{aligned}
\Lambda_\nu(ah)
&=Q_\nu(ah;t)
=\frac{V_\nu(B_\nu+ath)-V_\nu(B_\nu)}t\\
&=a\,\frac{V_\nu(B_\nu+(at)h)-V_\nu(B_\nu)}{at}\\
&=a\,Q_\nu(h;at)=a\Lambda_\nu(h).
\end{aligned}
\]
The case \(a=0\) follows from \(\Lambda_\nu(0)=0\).

For additivity, we will express the step in direction \(h+k\)
as the midpoint of steps in directions \(h\) and \(k\).
Given \(h,k\in V\), choose
\[
0<t<\frac12\min\{r_\tau(h),r_\tau(k)\}.
\]
This ensures \(2t\in I_\tau(h)\cap I_\tau(k)\). In particular,
each direction uses less than half of each admissibility margin:
\[
t\|v'\|_\infty<\frac{1-\tau}{2},\qquad
t\|v''\|_\infty<\frac{\tau a_q}{2}
\qquad(v\in\{h,k\}).
\]
The triangle inequality then gives
\[
\begin{aligned}
t\|(h+k)'\|_\infty
&\le t\|h'\|_\infty+t\|k'\|_\infty
<\frac{1-\tau}{2}+\frac{1-\tau}{2}=1-\tau,\\
t\|(h+k)''\|_\infty
&\le t\|h''\|_\infty+t\|k''\|_\infty
<\frac{\tau a_q}{2}+\frac{\tau a_q}{2}=\tau a_q.
\end{aligned}
\]
Thus \(t\in I_\tau(h+k)\), and all three perturbed curves below
belong to \(X_\nu\). Affinity gives
\[
\begin{aligned}
B_\nu+t(h+k)
&=\tfrac12(B_\nu+2th)+\tfrac12(B_\nu+2tk),\\
V_\nu(B_\nu+t(h+k))
&=\tfrac12V_\nu(B_\nu+2th)+\tfrac12V_\nu(B_\nu+2tk).
\end{aligned}
\]
Subtract \(V_\nu(B_\nu)\) and use the increment formula from
part~(iii) in each direction:
\[
\begin{aligned}
t\Lambda_\nu(h+k)
&=\tfrac12\bigl(V_\nu(B_\nu+2th)-V_\nu(B_\nu)\bigr)\\
&\quad+\tfrac12\bigl(V_\nu(B_\nu+2tk)-V_\nu(B_\nu)\bigr)\\
&=\tfrac12(2t\Lambda_\nu(h))+\tfrac12(2t\Lambda_\nu(k))\\
&=t\bigl(\Lambda_\nu(h)+\Lambda_\nu(k)\bigr).
\end{aligned}
\]
Dividing by \(t>0\) proves
\(\Lambda_\nu(h+k)=\Lambda_\nu(h)+\Lambda_\nu(k)\).
Together with homogeneity, this proves linearity on \(V\).

\noindent\emph{(v) Recovery of the representatives and uniqueness.}
For \(P\in X_\nu\), the difference \(P-B_\nu\) is smooth
and vanishes at both endpoints, hence belongs to \(V\). Take
\(0<t<\min\{1,r_\tau(P-B_\nu)\}\). Then
\[
B_\nu+t(P-B_\nu)=(1-t)B_\nu+tP.
\]
The choice of \(t\) permits both the increment formula of part~(iii)
and affinity on the segment joining \(B_\nu\) to \(P\).
Only the intermediate point needs the strict margins of part~(ii);
\(P\) is any member of \(X_\nu\). In particular, on the primal
side this includes smooth Lorenz curves with horizontal segments.
Thus
\[
\begin{aligned}
t\Lambda_\nu(P-B_\nu)
&=V_\nu(B_\nu+t(P-B_\nu))-V_\nu(B_\nu)\\
&=(1-t)V_\nu(B_\nu)+tV_\nu(P)-V_\nu(B_\nu)\\
&=t\bigl(V_\nu(P)-V_\nu(B_\nu)\bigr).
\end{aligned}
\]
Division by \(t>0\) gives
\(\Lambda_\nu(P-B_\nu)=V_\nu(P)-V_\nu(B_\nu)\).
Since \(E\in X_\nu\), subtraction at \(P\) and \(E\), using
linearity from part~(iv), gives
\begin{equation}
\begin{aligned}
\Lambda_\nu(P-E)
&=\Lambda_\nu\bigl((P-B_\nu)-(E-B_\nu)\bigr)\\
&=\Lambda_\nu(P-B_\nu)-\Lambda_\nu(E-B_\nu)\\
&=\bigl(V_\nu(P)-V_\nu(B_\nu)\bigr)
-\bigl(V_\nu(E)-V_\nu(B_\nu)\bigr)\\
&=V_\nu(P)-V_\nu(E)\qquad(P\in X_\nu).
\end{aligned}
\label{eq:generic-linear-part}
\end{equation}
If a linear \(\Gamma_\nu\) also satisfies this identity,
then, for \(0\ne t\in I_\tau(h)\),
\[
\begin{aligned}
t\Gamma_\nu(h)
&=\Gamma_\nu(B_\nu+th-E)
-\Gamma_\nu(B_\nu-E)\\
&=V_\nu(B_\nu+th)-V_\nu(B_\nu)
=t\Lambda_\nu(h).
\end{aligned}
\]
Since \(t\ne0\), division by \(t\) gives
\(\Gamma_\nu(h)=\Lambda_\nu(h)\) for every \(h\in V\).
Hence \(\Gamma_\nu=\Lambda_\nu\). Set
\[
T_0:=\Lambda_{-1},\qquad m_0:=\Lambda_{+1}.
\]
For \(F\in X_{+1}\), the sign in \(V_{+1}=-W_Y\) gives
\[
m_0(F-E)=V_{+1}(F)-V_{+1}(E)=W_Y(E)-W_Y(F).
\]
Thus equation~\eqref{eq:generic-linear-part} is exactly
\eqref{eq:original-primal-direction}--\eqref{eq:original-dual-direction}.
Uniqueness also implies independence of the choice of \(\tau\):
every such construction satisfies \eqref{eq:generic-linear-part},
whose statement does not involve \(\tau\).

\noindent\emph{(vi) Positivity and the weighted bounds.}
For \(h\in V\) with \(h\ge0\), choose \(0<t<r_\tau(h)\).
Both \(B_\nu+th\)
and \(B_\nu\) belong to \(X_\nu\) by part~(ii), and
\(B_\nu+th\ge B_\nu\). Thus \eqref{eq:two-chart-positivity} gives
\[
\Lambda_\nu(h)
=\frac{V_\nu(B_\nu+th)-V_\nu(B_\nu)}t\ge0.
\]
Now let \(h\in V\) be arbitrary. Lemma~\ref{lem:quadratic-division} gives
\[
M_h:=\|h/q\|_\infty<\infty,\qquad
M_hq-h\ge0,\qquad M_hq+h\ge0.
\]
Here \(h/q\) denotes its smooth extension at the endpoints.
The two inequalities follow from \(|h|\le M_hq\) on \((0,1)\)
and from \(h=q=0\) at the endpoints. Both \(M_hq-h\) and
\(M_hq+h\) belong to \(V\).
Applying positivity to \(M_hq-h\) and \(M_hq+h\), and then
linearity, gives
\[
0\le M_h\Lambda_\nu(q)-\Lambda_\nu(h),\qquad
0\le M_h\Lambda_\nu(q)+\Lambda_\nu(h).
\]
The two inequalities provide the upper and lower bounds
\[
-M_h\Lambda_\nu(q)\le\Lambda_\nu(h)\le M_h\Lambda_\nu(q).
\]
Taking \(\nu=-1\) and \(\nu=+1\), respectively, yields
\begin{equation}
|T_0(h)|\le M_hT_0(q),\qquad
|m_0(h)|\le M_hm_0(q).
\label{eq:unnormalized-weighted-bound}
\end{equation}

\noindent\emph{(vii) Strict positivity of the two scales.}
To normalize at \(q\), its value must be strictly positive for both
linear parts. A zero value would force the entire corresponding
linear part to vanish by the weighted bound. Smooth density rules
this out on the primal side; a smooth inverse pair then transfers
a strict comparison to the dual side.

Since \(q\ge0\), part~(vi) gives \(T_0(q),m_0(q)\ge0\),
so it suffices to exclude zero in each case.
If \(T_0(q)=0\), then \eqref{eq:unnormalized-weighted-bound}
gives \(T_0=0\). For \(L\in\mathcal{L}\), put
\(L_n=B_nL\). Lemma~\ref{lem:bernstein-density} gives
\(L_n\in X_{-1}\), so the recovered primal identity applies
to every approximant. Equations~\eqref{eq:original-primal-direction}
and~\eqref{eq:bernstein-uniform-bound} give
\[
U(L_n)=U(E),\qquad \|L_n-L\|_\infty\longrightarrow0.
\]
For each \(\varepsilon>0\), choose \(\delta\) from
\eqref{eq:representative-continuity} at \(L\), and then
\(n\) with \(\|L_n-L\|_\infty<\delta\). Thus
\[
|U(L)-U(E)|=|U(L)-U(L_n)|<\varepsilon.
\]
It follows that \(U(L)=U(E)\) for all \(L\), contradicting
nonconstancy. Hence \(\sigma_U>0\).

If \(m_0(q)=0\), then \(m_0=0\). Fix \(t\in(0,1)\), and put
\(R_t:=E-tq\). We have
\[
R_t(0)=0,\quad R_t(1)=1,\quad
R_t'=1-tq'\ge1-t>0,\quad R_t''=-tq''\ge ta_q>0.
\]
Thus \(R_t\in\mathcal{L}^\uparrow\cap C^\infty([0,1])\), and
\(G_t:=R_t^{-1}\in Y\) by Lemma~\ref{lem:domain-inversion}.

To apply the dual linear-part identity, we also verify
\(G_t\in C^\infty([0,1])\). For \(v\ne0\) with \(y,y+v\in[0,1]\),
write \(x=G_t(y)\), \(x_v=G_t(y+v)\). The inverse identity
and the fundamental theorem of calculus give
\[
v=R_t(x_v)-R_t(x)
=(x_v-x)\int_0^1R_t'(x+r(x_v-x))\,dr.
\]
The integral is at least \(1-t>0\). Dividing this equality gives
\[
\frac{x_v-x}{v}
=\frac{1}{\displaystyle\int_0^1R_t'(x+r(x_v-x))\,dr}
\longrightarrow\frac1{R_t'(G_t(y))}.
\]
The limit follows from continuity and \(R_t'\ge1-t>0\);
it also applies one-sidedly at the endpoints. Hence
\(G_t'(y)=1/R_t'(G_t(y))\).
Define smooth functions recursively by
\[
p_0(x):=x,\qquad p_{j+1}(x):=\frac{p_j'(x)}{R_t'(x)}.
\]
The denominator is bounded away from zero. Induction by the chain
rule gives \(G_t^{(j)}(y)=p_j(G_t(y))\) for every \(j\ge0\),
with continuous one-sided extensions at both endpoints.
Therefore \(G_t\in Y\cap C^\infty([0,1])\).

Equations~\eqref{eq:original-primal-direction}--\eqref{eq:original-dual-direction}
then imply
\[
D_U(R_t)=t\sigma_U>0=D_U(E),\qquad
D_W(R_t)=m_0(G_t-E)=0=D_W(E).
\]
Since \(E,R_t\in\mathcal{L}^\uparrow\), the strict primal
inequality would force \(D_W(E)<D_W(R_t)\) by
\eqref{eq:common-deficiency-order}, contrary to their equality.
Thus \(\sigma_W>0\). Recall that
\(T=T_0/\sigma_U\), \(m=m_0/\sigma_W\), and
\(\sigma_U=T_0(q)\), \(\sigma_W=m_0(q)\).
For every \(h\in V\), division of
\eqref{eq:unnormalized-weighted-bound} gives
\[
\begin{aligned}
|T(h)|&=\frac{|T_0(h)|}{\sigma_U}
\le\|h/q\|_\infty\frac{T_0(q)}{\sigma_U}
=\|h/q\|_\infty,\\
|m(h)|&=\frac{|m_0(h)|}{\sigma_W}
\le\|h/q\|_\infty\frac{m_0(q)}{\sigma_W}
=\|h/q\|_\infty.
\end{aligned}
\]
Both absolute values have the same upper bound. Taking their
maximum and using
\(\|h/q\|_\infty\le C_q\|h'\|_\infty\) from
\eqref{eq:quadratic-division-bound} proves
\eqref{eq:positive-weighted-bound}.
\end{proof}

Retain \(U,W_Y\) without rescaling. The constants
\(\sigma_U=T_0(q)\) and \(\sigma_W=m_0(q)\) are already fixed
by these representatives and the reference direction \(q\).
Their expressions in terms of reference curves follow from the
identities just proved. For any \(0<\tau<1\), the admissibility
checks in part~(ii) give
\(E-\tau q\in X_{-1}\) and \(E+\tau q\in X_{+1}\).
Equations~\eqref{eq:original-primal-direction}--\eqref{eq:original-dual-direction}
therefore give
\[
\begin{aligned}
U(E-\tau q)-U(E)&=T_0(-\tau q)=-\tau\sigma_U,\\
W_Y(E)-W_Y(E+\tau q)&=m_0(\tau q)=\tau\sigma_W.
\end{aligned}
\]
Rearranging and dividing by \(\tau>0\) yields
\[
\sigma_U=\frac{U(E)-U(E-\tau q)}{\tau},\qquad
\sigma_W=\frac{W_Y(E)-W_Y(E+\tau q)}{\tau}.
\]
Each quotient is independent of \(\tau\): it equals the
corresponding fixed constant \(\sigma_U\) or \(\sigma_W\).

Use these two positive constants to define the normalized gaps:
\[
\begin{aligned}
d_U(L)&:=\frac{D_U(L)}{\sigma_U}
&&(L\in\mathcal{L}),\\
d_W(L)&:=\frac{D_W(L)}{\sigma_W}
&&(L\in\mathcal{L}^\uparrow).
\end{aligned}
\]
These definitions apply on the full domains displayed above.
To express the gaps through \(T\) and \(m\), we use the smoothness
conditions in the recovered linear identities.
For \(L\in\mathcal{L}\cap C^\infty([0,1])\), the direction
\(E-L\) belongs to \(V\), and
\[
D_U(L)=U(E)-U(L)=-T_0(L-E)=T_0(E-L).
\]
For \(L\in\mathcal{L}^\uparrow\) with smooth inverse,
\(L^{-1}-E\in V\). Since \(W(L)=W_Y(L^{-1})\) and
\(E^{-1}=E\), the dual identity gives
\[
D_W(L)=W_Y(E)-W_Y(L^{-1})=m_0(L^{-1}-E).
\]
Division by \(\sigma_U\) and \(\sigma_W\), respectively, now yields
\begin{equation}
\begin{aligned}
d_U(L)&=T(E-L)&&(L\in\mathcal{L}\cap C^\infty([0,1])),\\
d_W(L)&=m(L^{-1}-E)
&&(L\in\mathcal{L}^\uparrow,\ L^{-1}\in C^\infty([0,1])).
\end{aligned}
\label{eq:normalized-directional-gaps}
\end{equation}

The normalized gaps also retain the common ordering on
\(\mathcal{L}^\uparrow\). For any \(L,M\) in that domain,
\[
\begin{aligned}
d_U(L)-d_U(M)&=\frac{D_U(L)-D_U(M)}{\sigma_U},\\
d_W(L)-d_W(M)&=\frac{D_W(L)-D_W(M)}{\sigma_W}.
\end{aligned}
\]
Since the denominators of the above two equations are both positive, each difference has the same sign
as its numerator. Apply \eqref{eq:common-deficiency-order} therefore yields:
\begin{equation}
\begin{aligned}
d_U(L)=d_U(M)&\Longleftrightarrow d_W(L)=d_W(M),\\
d_U(L)<d_U(M)&\Longleftrightarrow d_W(L)<d_W(M).
\end{aligned}
\label{eq:normalized-fibres}
\end{equation}
Thus the two normalized gaps have the same level sets and order
those sets in the same way. This conclusion holds for all
\(L,M\in\mathcal{L}^\uparrow\), by the definitions of the gaps.

Finally, \eqref{eq:positive-weighted-bound} controls the effect of
an error \(h\in V\) on either linear functional through
\(C_q\|h'\|_\infty\). The next subsection therefore estimates
the first spatial derivative of the inversion remainder along
with the remainder itself.

\subsection{Inversion with integral remainder}
\label{subsec:inverse-variation}

The purpose of this subsection is to obtain the second-order
Taylor expansion needed for the subsequent coefficient comparisons.
For a fixed \(g\in V\), choose \(\rho>0\) with
\(\rho\|g'\|_\infty<1\). We aim to prove
\[
\begin{aligned}
(E-tg)^{-1}&=E+tg+t^2gg'+R_g(t,\cdot),\\
\|R_g(t,\cdot)\|_{C^1}&\le C(g,\rho)|t|^3
\qquad(|t|\le\rho).
\end{aligned}
\]
Here \((E-tg)^{-1}\) denotes the ordinary inverse function.

The proof below is constituted of arguments that are routine symbol pushing, they are included so this paper will be self-contained. We will first establish the desired smooth invertibility, then evaluate the two coefficients at the zero point and the boundedness property of the integral remainder on a fixed parameter interval. This calculation concerns smooth increasing bijections. The next subsection checks the additional Lorenz shape conditions before making preference comparisons.

\begin{lemma}[Second-order inversion with a uniform \(C^1\) remainder]
\label{lem:inverse-variation}
Fix \(g\in V\), so \(g\) is smooth and \(g(0)=g(1)=0\).
The derivative of \(x\mapsto x-tg(x)\) determines the admissible
parameter interval. Define
\[
\Delta_g(t,x):=1-tg'(x),\qquad
\mu_g(t):=\min_{x\in[0,1]}\Delta_g(t,x),\qquad
I_g:=\{t\in\mathbb R:\mu_g(t)>0\}.
\]
Then \(I_g\) is an open interval containing \(0\).
For each \(t\in I_g\), the map \(x\mapsto x-tg(x)\) is a strictly
increasing bijection of \([0,1]\). Write
\(F_g(t,\cdot):=(E-tg)^{-1}\) for its ordinary inverse function.
The function \(F_g(t,y)\) is smooth jointly in \((t,y)\) on
\(I_g\times[0,1]\), with one-sided spatial derivatives at the endpoints.
The following expansion is in the parameter \(t\) at zero,
with \(y\) fixed. For every \(t\in I_g\),
\begin{equation}
F_g(t,\cdot)=E+tg+t^2gg'+R_g(t,\cdot),
\qquad R_g(t,\cdot)\in V,
\label{eq:inverse-taylor}
\end{equation}
with the exact remainder
\begin{equation}
R_g(t,y)
=\frac{t^3}{2}\int_0^1(1-r)^2\partial_t^3F_g(rt,y)\,dr.
\label{eq:inverse-integral-remainder}
\end{equation}

For \(f\in C^1([0,1])\), use the norm
\[
\|f\|_{C^1}:=\max\{\|f\|_\infty,\|f'\|_\infty\}.
\]
A bound for this maximum bounds both the function and its first
derivative. In \(\|R_g(t,\cdot)\|_{C^1}\), the parameter \(t\)
is held fixed: both suprema and the derivative are taken in \(y\).
For any compact interval \(K\) with \(0\in K\subset I_g\),
define
\begin{equation}
C(g,K):=\frac16
\max_{\substack{s\in K\\y\in[0,1]}}
\max\{|\partial_t^3F_g(s,y)|,\,
|\partial_y\partial_t^3F_g(s,y)|\}<\infty.
\label{eq:inverse-remainder-constant}
\end{equation}
Then the same constant works for every \(t\in K\):
\begin{align}
\|R_g(t,\cdot)\|_{C^1}
&\le C(g,K)|t|^3,
\label{eq:inverse-C1-remainder}\\
\max\{|T(R_g(t,\cdot))|,
|m(R_g(t,\cdot))|\}
&\le C_q C(g,K)|t|^3.
\label{eq:inverse-scalar-remainder}
\end{align}
In particular, if \(\rho>0\) and \(\rho\|g'\|_\infty<1\), we
write \(C(g,\rho):=C(g,[-\rho,\rho])\). The estimates then hold
uniformly for \(|t|\le\rho\).
\end{lemma}
\begin{proof}
\noindent\emph{(i) A positive slope gives an inverse and joint continuity.}
For \(s,t\in\mathbb R\), the difference between the two minima
is at most the largest pointwise difference. Hence
\[
\begin{aligned}
|\mu_g(t)-\mu_g(s)|
&\le\max_{x\in[0,1]}|\Delta_g(t,x)-\Delta_g(s,x)|\\
&=|t-s|\|g'\|_\infty,\qquad \mu_g(0)=1.
\end{aligned}
\]
Thus \(I_g\) is open and contains zero. For
\(s,t\in I_g\) and \(\lambda\in[0,1]\),
\[
\begin{aligned}
\Delta_g((1-\lambda)s+\lambda t,x)
&=(1-\lambda)\Delta_g(s,x)+\lambda\Delta_g(t,x)\\
&\ge(1-\lambda)\mu_g(s)+\lambda\mu_g(t)>0.
\end{aligned}
\]
Hence \(I_g\) is an interval. Its endpoints can be read directly
from the extreme slopes. Put
\[
m_g:=\min_{[0,1]}g',\qquad M_g:=\max_{[0,1]}g'.
\]
Since \(\int_0^1g'=g(1)-g(0)=0\), either \(g=0\), or
\(m_g<0<M_g\). Indeed, a continuous derivative of only one sign
with zero integral must vanish identically, and the endpoint
condition would then give \(g=0\). Therefore
\begin{equation}
\begin{aligned}
\mu_g(t)&=
\begin{cases}
1-tM_g,&t\ge0,\\
1-tm_g,&t\le0,
\end{cases}\\
I_g&=
\begin{cases}
\mathbb R,&g=0,\\
(1/m_g,1/M_g),&g\ne0.
\end{cases}
\end{aligned}
\label{eq:inverse-exact-parameter-domain}
\end{equation}
Thus \(I_g\) is precisely the parameter interval on which the
slope is strictly positive everywhere. If \(t\in I_g\),
the map \(f_t:=E-tg\) fixes both endpoints and satisfies
\[
f_t(b)-f_t(a)=\int_a^b\Delta_g(t,x)\,dx
\ge\mu_g(t)(b-a)>0\qquad(0\le a<b\le1).
\]
Continuity and the endpoint values therefore make \(f_t\) an
increasing bijection of \([0,1]\). Its inverse \(F_g(t,\cdot)\)
is characterized by \((E-tg)(F_g(t,y))=y\), or equivalently,
\begin{equation}
F_g(t,y)-t g(F_g(t,y))=y,\qquad
F_g(t,0)=0,\quad F_g(t,1)=1.
\label{eq:inverse-defining-identity}
\end{equation}

We will differentiate this identity only after establishing
continuity of the inverse in both variables. Write
\(a=F_g(t,y)\), \(b=F_g(s,z)\). The two inverse relations give
\[
\begin{aligned}
a-tg(a)&=y,\qquad b-sg(b)=z,\\
(a-b)-t\bigl(g(a)-g(b)\bigr)&=y-z+(t-s)g(b).
\end{aligned}
\]
The fundamental theorem of calculus along the segment from \(b\)
to \(a\) gives
\[
g(a)-g(b)=(a-b)\int_0^1g'(b+r(a-b))\,dr,
\]
including when \(a=b\). Substitution yields
\begin{equation}
(a-b)\left(1-t\int_0^1g'(b+r(a-b))\,dr\right)
=y-z+(t-s)g(b).
\label{eq:inverse-difference-identity}
\end{equation}
The factor in parentheses is at least \(\mu_g(t)>0\), since
the segment from \(b\) to \(a\) lies in \([0,1]\). Consequently,
\begin{equation}
|F_g(t,y)-F_g(s,z)|
\le\frac{|y-z|+|t-s|\|g\|_\infty}{\mu_g(t)}.
\label{eq:inverse-joint-modulus}
\end{equation}
With \((t,y)\) fixed, this bound tends to zero as \((s,z)\to(t,y)\);
it proves joint continuity, including at the spatial endpoints.

\smallskip
\noindent\emph{(ii) Differentiability, including at the endpoints.}
Fix \((t,y)\in I_g\times[0,1]\), abbreviate \(F=F_g(t,y)\),
and put \(d=\Delta_g(t,F)>0\). Choose nonzero increments \(v,u\)
with \(y+v\in[0,1]\) and \(t+u\in I_g\). Let
\(F_v=F_g(t,y+v)\), \(F_u=F_g(t+u,y)\).
Apply \eqref{eq:inverse-difference-identity} first to
\((t,y+v),(t,y)\), and then to \((t+u,y),(t,y)\).
Dividing the resulting identities by \(v\) and \(u\), respectively,
gives
\[
\begin{aligned}
\frac{F_v-F}{v}
&=\frac{1}{1-t\int_0^1g'(F+r(F_v-F))\,dr},\\
\frac{F_u-F}{u}
&=\frac{g(F)}
{1-(t+u)\int_0^1g'(F+r(F_u-F))\,dr}.
\end{aligned}
\]
Joint continuity gives \(F_v,F_u\to F\).
Uniform continuity of \(g'\) on \([0,1]\) makes both denominators
converge to \(d>0\). Thus
\begin{equation}
\partial_yF_g(t,y)=\frac1{\Delta_g(t,F_g(t,y))},\qquad
\partial_tF_g(t,y)=\frac{g(F_g(t,y))}
{\Delta_g(t,F_g(t,y))}.
\label{eq:inverse-first-derivatives}
\end{equation}
These derivatives are continuous. At \(y=0,1\), the spatial
quotients use increments inside \([0,1]\), so the same formulas
give the one-sided derivatives.

To justify subsequent differentiations and their uniform bounds,
we record the structure of all higher derivatives. Suppress the
subscript on \(\Delta_g\). The index \(w\) records a finite sequence
of differentiations: for example, \(w=ttty\) means three parameter
derivatives followed by one spatial derivative.
The auxiliary functions \(P_w(t,x)\) are calculated with \(t,x\)
independent; we will then substitute \(x=F_g(t,y)\). Define
\begin{equation}
\begin{aligned}
P_{\varnothing}(t,x)&:=x,\\
P_{wt}(t,x)&:=\partial_tP_w(t,x)
+\frac{g(x)}{\Delta(t,x)}\partial_xP_w(t,x),\\
P_{wy}(t,x)&:=\frac1{\Delta(t,x)}\partial_xP_w(t,x).
\end{aligned}
\label{eq:inverse-derivative-recurrence}
\end{equation}
The two recursion rules are the chain rule for parameter and
spatial differentiation, respectively, after substituting
\(x=F_g(t,y)\). If \(D_w\) denotes successive differentiation
in the order recorded by \(w\), induction using
\eqref{eq:inverse-first-derivatives} gives
\[
D_wF_g(t,y)=P_w(t,F_g(t,y))
\qquad(t\in I_g,\ 0<y<1).
\]
For the first derivatives the recursion gives
\(P_t=g/\Delta\) and \(P_y=1/\Delta\), in agreement with
\eqref{eq:inverse-first-derivatives}. For the induction step,
write \(F=F_g(t,y)\). The chain rule gives
\[
\begin{aligned}
\partial_t\bigl(P_w(t,F)\bigr)
&=(\partial_tP_w)(t,F)+(\partial_xP_w)(t,F)\,\partial_tF\\
&=(\partial_tP_w)(t,F)
+(\partial_xP_w)(t,F)\frac{g(F)}{\Delta(t,F)}
=P_{wt}(t,F),\\
\partial_y\bigl(P_w(t,F)\bigr)
&=(\partial_xP_w)(t,F)\,\partial_yF\\
&=\frac{(\partial_xP_w)(t,F)}{\Delta(t,F)}
=P_{wy}(t,F).
\end{aligned}
\]

Every \(P_w\) is a finite sum of smooth numerators divided by
powers of \(\Delta\). This property is preserved by both recursion
rules, as is seen from
\[
\begin{aligned}
\partial_t\!\left(\frac{a}{\Delta^j}\right)
&=\frac{a_t}{\Delta^j}
+\frac{jag'}{\Delta^{j+1}},\\
\partial_x\!\left(\frac{a}{\Delta^j}\right)
&=\frac{a_x}{\Delta^j}
+\frac{jta g''}{\Delta^{j+1}}
\qquad(j\ge0).
\end{aligned}
\]
Thus these expressions extend continuously to \(x=0,1\)
wherever \(\Delta>0\). Their compositions with \(F_g\) give
continuous extensions of all interior derivatives to \(y=0,1\).

These extensions are the required endpoint derivatives.
For example, the fundamental theorem of calculus, first on an
interior interval and then by continuity at zero, gives
\[
\frac{P_w(t,F_g(t,v))-P_w(t,F_g(t,0))}{v}
=\int_0^1P_{wy}(t,F_g(t,rv))\,dr
\longrightarrow P_{wy}(t,F_g(t,0))
\quad(v\downarrow0).
\]
At the endpoint \(1\), use the increment \(-v\) and the path
\(1-rv\); the same argument gives \(P_{wy}(t,F_g(t,1))\).
For parameter derivatives at an endpoint \(p\in\{0,1\}\),
use \(F_g(t,p)=p\): differentiation of \(P_w(t,p)\) agrees
with \(P_{wt}(t,p)\), since \(g(p)=0\).
This completes the smoothness argument on the closed spatial
interval.

\smallskip
\noindent\emph{(iii) The linear and quadratic Taylor coefficients.}
The smoothness established in part~(ii) now justifies differentiating
the inverse identity with respect to \(t\), while holding \(y\) fixed.
We compute the derivatives at zero that enter the second-order
expansion. Write \(F=F_g(t,y)\); the defining identity is
\(F-tg(F)=y\). The product and chain rules give
\[
\begin{aligned}
\partial_tF-g(F)-t g'(F)\partial_tF&=0,\\
(1-tg'(F))\partial_tF&=g(F).
\end{aligned}
\]
Here \(g'\) and \(g''\) denote derivatives of \(g\) with respect
to its argument. To differentiate the second identity once more,
we must also differentiate the factor \(1-tg'(F)\), since \(F\)
depends on \(t\):
\[
\partial_t\bigl(1-tg'(F)\bigr)
=-g'(F)-t g''(F)\partial_tF.
\]
Applying the product rule to the left-hand side and the chain rule
to the right-hand side therefore yields
\[
\begin{aligned}
&(1-tg'(F))\partial_t^2F
-g'(F)\partial_tF-t g''(F)(\partial_tF)^2\\
&\qquad=g'(F)\partial_tF.
\end{aligned}
\]
Moving the last two terms on the left to the right combines the two
copies of \(g'(F)\partial_tF\), giving
\[
(1-tg'(F))\partial_t^2F
=2g'(F)\partial_tF+t g''(F)(\partial_tF)^2.
\]
At \(t=0\), the inverse identity gives \(F_g(0,y)=y\).
The factor \(1-tg'(F)\) is then one, and the term multiplied by
\(t\) in the second-derivative identity vanishes. Substituting first
into the first-derivative identity and then into the second gives
\begin{equation}
\begin{aligned}
\partial_tF_g(0,y)
&=g(F_g(0,y))=g(y),\\
\partial_t^2F_g(0,y)
&=2g'(F_g(0,y))\partial_tF_g(0,y)\\
&=2g'(y)g(y).
\end{aligned}
\label{eq:inverse-derivatives-at-zero}
\end{equation}
Thus the linear and quadratic Taylor terms are \(t g(y)\) and
\(\tfrac{t^2}{2}\partial_t^2F_g(0,y)=t^2g(y)g'(y)\), respectively.

The two coefficients will play the central role in the rest of whole proof project.
The coefficient \(g\) will identify the normalized linear
functionals; the coefficient \(gg'\) will yield the quadratic
identity. The remainder estimate justifies both
coefficient comparisons.

It remains to derive the integral remainder and prove a bound
that holds uniformly on a fixed parameter interval.

\smallskip
\noindent\emph{(iv) One constant on a fixed parameter interval.}
Fix a compact interval \(K=[a,b]\) with \(a\le0\le b\) and
\(K\subset I_g\).
Its slope margin
\[
\begin{aligned}
\delta(g,K)
&:=\min_{\substack{s\in K\\x\in[0,1]}}\Delta_g(s,x)
=\min_{s\in K}\mu_g(s)\\
&=\min\{1-a m_g,\,1-b M_g\}>0
\end{aligned}
\]
follows from \eqref{eq:inverse-exact-parameter-domain}. In particular,
on a one-sided interval it is exactly
\[
\delta(g,[0,b])=1-bM_g,
\qquad 0\le b<1/M_g\quad(g\ne0).
\]
The recurrence in
part~(ii) shows that the two derivatives required for a \(C^1\)
cubic remainder are
\[
\partial_t^3F_g(s,y)=P_{ttt}(s,F_g(s,y)),\qquad
\partial_y\partial_t^3F_g(s,y)=P_{ttty}(s,F_g(s,y)).
\]
Their numerators are bounded on \(K\times[0,1]\), and every
denominator is bounded away from zero by a power of \(\delta(g,K)\).
Since \(F_g(s,\cdot)\) maps \([0,1]\) onto itself,
\[
C(g,K)=\frac16
\max_{\substack{s\in K\\x\in[0,1]}}
\max\{|P_{ttt}(s,x)|,|P_{ttty}(s,x)|\}<\infty.
\]
This proves the finiteness claimed in
\eqref{eq:inverse-remainder-constant}. By its definition, the
same number bounds both derivatives at every point of the rectangle:
\[
\begin{aligned}
|\partial_t^3F_g(s,y)|&\le6C(g,K),\\
|\partial_y\partial_t^3F_g(s,y)|&\le6C(g,K)
\qquad(s\in K,\ y\in[0,1]).
\end{aligned}
\]
The constant is determined by \(g\) and the fixed interval \(K\).
It remains unchanged when the evaluation point \((s,y)\) varies.

For later comparisons, choose \(\rho>0\) with
\(\rho\|g'\|_\infty<1\). Such a radius always exists, and
\[
\delta(g,[-\rho,\rho])
=1-\rho\max\{|m_g|,|M_g|\}
=1-\rho\|g'\|_\infty>0.
\]
We may therefore define
\begin{equation}
C(g,\rho):=C(g,[-\rho,\rho])<\infty.
\label{eq:compact-third-derivative-bound}
\end{equation}
The radius is fixed before \(t\) varies. More generally, for any
\(t\in I_g\), the interval
\(K_t:=[\min\{0,t\},\max\{0,t\}]\) is compact and contained in
\(I_g\), so \(C(g,K_t)\) is finite. The symmetric-radius
estimate supplies uniformity near zero while the intervals \(K_t\)
retain the expansion on the full parameter domain.

\smallskip
\noindent\emph{(v) The integral remainder and the resulting scalar error.}
For fixed \(y\), Taylor's formula in the parameter follows by
integrating \(\partial_t^3F_g\) three times. Explicitly, two
integrations by parts give
\[
\begin{aligned}
&\frac12\int_0^t(t-s)^2\partial_t^3F_g(s,y)\,ds\\
&\quad=-\frac{t^2}{2}\partial_t^2F_g(0,y)
+\int_0^t(t-s)\partial_t^2F_g(s,y)\,ds\\
&\quad=F_g(t,y)-F_g(0,y)-t\partial_tF_g(0,y)
-\frac{t^2}{2}\partial_t^2F_g(0,y).
\end{aligned}
\]
Oriented integrals give the same identity for \(t<0\).
Using \eqref{eq:inverse-derivatives-at-zero} and \(F_g(0,y)=y\),
the last expression is exactly
\(F_g(t,y)-y-tg(y)-t^2g(y)g'(y)\).
Thus this difference defines the remainder \(R_g(t,y)\).
In the integral, the substitution \(s=rt\) gives
\((t-s)^2=t^2(1-r)^2\) and \(ds=t\,dr\), so
\[
R_g(t,y)
=\frac12\int_0^t(t-s)^2\partial_t^3F_g(s,y)\,ds
=\frac{t^3}{2}\int_0^1(1-r)^2\partial_t^3F_g(rt,y)\,dr.
\]
This proves \eqref{eq:inverse-taylor} and
\eqref{eq:inverse-integral-remainder}, also at \(t=0\).
Here \(\partial_t^3F_g(rt,y)\) denotes the third derivative in
the first argument of \(F_g\), evaluated at \((rt,y)\).

We need the same control after spatial differentiation.
The continuous mixed derivatives established above permit
Lemma~\ref{lem:parameter-integration} to be applied in \(y\), giving
\[
\partial_yR_g(t,y)
=\frac{t^3}{2}\int_0^1(1-r)^2
\partial_y\partial_t^3F_g(rt,y)\,dr.
\]
Fix \(t\in K\). Since \(K\) is an interval containing zero,
\(rt\in K\) for every \(0\le r\le1\). The two derivative bounds
from part~(iv) therefore apply throughout both integrals.
For every \(y\in[0,1]\), they give
\[
\begin{aligned}
|R_g(t,y)|
&\le\frac{|t|^3}{2}\int_0^1(1-r)^2
|\partial_t^3F_g(rt,y)|\,dr\\
&\le\frac{|t|^3}{2}\,6C(g,K)\int_0^1(1-r)^2\,dr
=C(g,K)|t|^3,\\
|\partial_yR_g(t,y)|
&\le\frac{|t|^3}{2}\int_0^1(1-r)^2
|\partial_y\partial_t^3F_g(rt,y)|\,dr\\
&\le\frac{|t|^3}{2}\,6C(g,K)\int_0^1(1-r)^2\,dr
=C(g,K)|t|^3.
\end{aligned}
\]
Here \(\int_0^1(1-r)^2\,dr=1/3\), which explains the factor
\(1/6\) in the definition of \(C(g,K)\).
The right-hand sides are independent of \(y\). Taking suprema
and then the maximum of the two bounds yields
\[
\begin{aligned}
\|R_g(t,\cdot)\|_{C^1}
&=\max\{\|R_g(t,\cdot)\|_\infty,\,
\|\partial_yR_g(t,\cdot)\|_\infty\}\\
&\le C(g,K)|t|^3.
\end{aligned}
\]
This proves \eqref{eq:inverse-C1-remainder} for every \(t\in K\)
with the same constant.
Moreover, the expansion defines a smooth function of \(y\),
and \(F_g(t,p)=p\), \(g(p)=0\) for \(p=0,1\) give
\(R_g(t,0)=R_g(t,1)=0\).
Hence \(R_g(t,\cdot)\in V\).

We can now use \eqref{eq:positive-weighted-bound}, since its input
\(h=R_g(t,\cdot)\) belongs to \(V\). For fixed \(t\), the derivative
\(h'\) appearing in that bound is 
\(\partial_yR_g(t,\cdot)\). Hence
\[
\begin{aligned}
\max\{|T(R_g(t,\cdot))|,
|m(R_g(t,\cdot))|\}
&\le C_q\|\partial_yR_g(t,\cdot)\|_\infty\\
&\le C_q C(g,K)|t|^3,
\end{aligned}
\]
where the last step uses the derivative estimate proved above.
This is \eqref{eq:inverse-scalar-remainder}.
This is the estimate needed in the next subsection: applying
either functional to the inverse expansion preserves a cubic
error with a constant fixed on the chosen parameter interval.
The estimate is applied to the remainder as a function; it
requires no interchange of a functional with an integral.
\end{proof}

\subsection{Coefficient comparison and the bilinear identity}
\label{subsec:second-order-compatibility}

This is where the two mixture structures impose a restriction on each other that helps avoid the single-sided blow-up counterexample. We will need some classical techniques that are standard in functional analysis, see \cite{rudin1991,conway1990,kreyszig1978,yosida1980}. For our application, these techniques are adjusted in terms so explicit  remainder estimates for inversion expansion could be made use of.

Given the two affine functionals represent the same ordering of Lorenz curve, We would
compare the first- and second-order coefficients of their gap
expansions by choosing curves with the same normalized gap on the primal side as
their dual gaps must agree as well. The functionals \(m,T\) and
the reference direction \(q\) remain fixed throughout this argument.
For each admissible direction \(g\), we obtain an identity valid
for every sufficiently small positive \(t\). This family of exact
identities forces agreement first of the linear coefficients and
then of the quadratic coefficients. Perturbations of \(q\) extend
these conclusions to all smooth directions, and polarization gives
the bilinear identity below.

\begin{proposition}[A bilinear identity for the normalized functional]
\label{prop:second-order-compatibility}
For \(h\in V\), the normalized functionals of
Proposition~\ref{prop:positive-directions}:
\[
m(h)=T(h).
\]
With \(c:=2m(qq')=2T(qq')\),
\begin{equation}
T((hk)')=c T(h)T(k)
\qquad(h,k\in V).
\label{eq:bilinear-compatibility}
\end{equation}
\end{proposition}
\begin{proof}
\noindent\emph{(i) The reference family and its equality fibres.}
For \(0\le s<s_q:=1/q'(0)\), set \(R_s:=E-sq\).
Since \(q(0)=q(1)=0\), one has \(R_s(0)=0\) and \(R_s(1)=1\).
The derivative bounds are
\[
\begin{aligned}
R_s'&\ge1-sq'(0)>0,\\
R_s''&=-sq''\ge sa_q\ge0.
\end{aligned}
\]
Thus \(R_s\in\mathcal{L}^\uparrow\cap C^\infty([0,1])\), and its
inverse satisfies
\[
R_s^{-1}=F_q(s,\cdot)\in Y\cap C^\infty([0,1]).
\]
The normalized primal-gap formula and \(T(q)=1\) give
\[
d_U(R_s)=sT(q)=s.
\]
Define
\begin{equation}
\Phi:[0,s_q)\to[0,\infty),\qquad
\Phi(s):=d_W(R_s)=m(F_q(s,\cdot)-E).
\label{eq:reference-transfer}
\end{equation}
The identity \(d_U(R_s)=s\) makes the reference family a convenient
choice of one curve for each primal gap in \([0,s_q)\). Thus \(\Phi\)
determines the dual gap of any curve with the same primal gap,
not just the dual gaps along the reference family.

For \(L\in\mathcal{L}^\uparrow\) and \(s=d_U(L)<s_q\),
\begin{equation}
d_U(L)=s=d_U(R_s)
\ \Longrightarrow\ d_W(L)=d_W(R_s)=\Phi(s).
\label{eq:all-curve-transfer}
\end{equation}
This is \eqref{eq:normalized-fibres}. It also gives
\[
0\le s<t<s_q\ \Longrightarrow\
d_U(R_s)<d_U(R_t)\ \Longrightarrow\ \Phi(s)<\Phi(t),
\qquad \Phi(0)=0.
\]
For \(s_0\in[0,s_q)\), \(\varepsilon>0\), continuity of \(W\)
at \(R_{s_0}\) gives \(\delta>0\) with
\[
L\in\mathcal{L}^\uparrow,\quad\|L-R_{s_0}\|_\infty<\delta
\ \Longrightarrow\ |W(L)-W(R_{s_0})|<\sigma_W\varepsilon.
\]
Since
\[
\|R_s-R_{s_0}\|_\infty=\|q\|_\infty|s-s_0|,
\qquad
|s-s_0|<\delta/\|q\|_\infty
\ \Longrightarrow\ |\Phi(s)-\Phi(s_0)|<\varepsilon,
\]
\(\Phi\) is continuous, including at zero.

\noindent\emph{(ii) Exact scalar expansions.}
Linear functional will be applied to the Taylor expansion of inversion.
The remainder estimate controls the resulting scalar error and supplies an
expansion of \(\Phi\) at zero; no smoothness of the transfer map
has been assumed.

For \(g\in V\) and \(t\in I_g\), let
\(r_g(t):=m(R_g(t,\cdot))\).
Here \(R_g(t,\cdot)\) is a function of the spatial variable,
whereas \(r_g(t)\) is a real number.
All terms lie in the domain \(V\) of \(m\): the inverse fixes
both endpoints, so \(F_g(t,\cdot)-E\in V\), and \(gg'\) vanishes
there because \(g\) does. Lemma~\ref{lem:inverse-variation}
also gives \(R_g(t,\cdot)\in V\).
For each fixed \(t\), the numbers \(t\) and \(t^2\) are scalars.
Thus \eqref{eq:inverse-taylor}, additivity, and homogeneity of \(m\)
give the exact identity
\begin{equation}
\begin{aligned}
m(F_g(t,\cdot)-E)
&=m\bigl(tg+t^2gg'+R_g(t,\cdot)\bigr)\\
&=t m(g)+t^2m(gg')+m(R_g(t,\cdot))\\
&=t m(g)+t^2m(gg')+r_g(t).
\end{aligned}
\label{eq:scalar-inverse-expansion}
\end{equation}
For every fixed \(\rho>0\) with \([0,\rho]\subset I_g\),
Lemma~\ref{lem:inverse-variation} supplies
\(C(g,[0,\rho])<\infty\).
To bound the scalar error, apply \eqref{eq:positive-weighted-bound}
to the function \(R_g(t,\cdot)\), then use
\eqref{eq:inverse-C1-remainder}:
\[
\begin{aligned}
|r_g(t)|
&=|m(R_g(t,\cdot))|\\
&\le C_q\|\partial_yR_g(t,\cdot)\|_\infty\\
&\le C_q\|R_g(t,\cdot)\|_{C^1}\\
&\le C_q C(g,[0,\rho])t^3
\qquad(0\le t\le\rho).
\end{aligned}
\]
The middle inequality uses the definition
\(\|f\|_{C^1}=\max\{\|f\|_\infty,\|f'\|_\infty\}\).
Thus the expansion holds on \(I_g\), and its error
is controlled by one fixed constant on each such smaller interval.

We now turn to the other direction of functional: \(g=q\).
Set \(c:=2m(qq')\). This constant is fixed by \(m\) and \(q\),
independently of the comparison direction \(g\).
Recall from \eqref{eq:reference-transfer} that
\(\Phi(s)=m(F_q(s,\cdot)-E)\).
Taking \(g=q\), \(t=s\) in \eqref{eq:scalar-inverse-expansion},
and using \(m(q)=1\) and \(m(qq')=c/2\), gives
\begin{equation}
\begin{aligned}
\Phi(s)
&=m(F_q(s,\cdot)-E)\\
&=s\,m(q)+s^2m(qq')+r_q(s)\\
&=s+\frac{c}{2}s^2+r_q(s)
\qquad(0\le s<s_q).
\end{aligned}
\label{eq:transfer-second-order}
\end{equation}
For every fixed \(0<\eta<s_q=1/q'(0)\), the derivative \(q'\)
is decreasing, so
\[
1-sq'(y)\ge1-\eta q'(0)>0
\qquad(0\le s\le\eta,\ 0\le y\le1).
\]
Hence \([0,\eta]\subset I_q\). Applying the preceding scalar
remainder bound with \(g=q\) and \(\rho=\eta\) gives
\[
|r_q(s)|\le C_q C(q,[0,\eta])s^3
\qquad(0\le s\le\eta).
\]
We will use this bound with \(\eta=\rho T(g)\), so the reference
expansion is controlled throughout the same comparison interval;
the required inequality \(0<\rho T(g)<s_q\) is verified in part~(iii).

\noindent\emph{(iii) Admissible one-sided comparisons.}
To make \(E-tg\) a Lorenz curve for small \(t>0\), we initially
restrict \(g\) to directions with uniformly strict concavity.
This also guarantees a positive primal gap. Later, the curvature
margin of \(q\) will let us recover identities for every smooth
direction from these admissible comparisons.

Define
\[
V_{\mathrm{conc}}:=\{g\in V:\gamma_g:=\min_{[0,1]}(-g'')>0\}.
\]
Fix an arbitrary \(g\in V_{\mathrm{conc}}\). To verify that \(T(g)>0\),
compare \(g\) with a positive multiple of \(q\):
\[
\lambda_g:=\min_{[0,1]}\frac{-g''}{-q''}>0,\qquad
w_g:=g-\lambda_gq.
\]
The endpoint values satisfy \(w_g(0)=w_g(1)=0\), and
\[
w_g''=g''-\lambda_gq''
=-(-q'')\left(\frac{-g''}{-q''}-\lambda_g\right)\le0.
\]
Concavity, by \eqref{eq:curvature-shape} applied to \(-w_g\), gives
\[
w_g(x)\ge(1-x)w_g(0)+xw_g(1)=0
\qquad(0\le x\le1).
\]
Thus \(g\ge\lambda_gq\). Positivity and normalization of \(T\) now give
\[
T(g)\ge\lambda_gT(q)=\lambda_g>0.
\]
The derivatives of \(g\) and \(q\) are decreasing, so their maximal
slopes are \(g'(0)\) and \(q'(0)\), respectively.
For this fixed \(g\), choose one radius \(\rho\) satisfying
\begin{equation}
0<\rho<\frac1{\max\{g'(0),q'(0)T(g)\}}.
\label{eq:comparison-radius}
\end{equation}
Keep both \(g\) and \(\rho\) fixed in the coefficient comparisons
below. Every \(t\) with \(0<t\le\rho\) then satisfies
\[
\begin{aligned}
(E-tg)'&\ge1-\rho g'(0)>0,\\
(E-tg)''&=-tg''\ge t\gamma_g>0.
\end{aligned}
\]
The endpoint values are \((E-tg)(0)=0\) and \((E-tg)(1)=1\).
Consequently,
\[
E-tg\in\mathcal{L}^\uparrow\cap C^\infty([0,1]),\qquad
F_g(t,\cdot)\in Y\cap C^\infty([0,1]).
\]
Since \(E-tg\) is a smooth Lorenz curve, its normalized primal gap
can be evaluated by \eqref{eq:normalized-directional-gaps}:
\begin{equation}
\begin{aligned}
d_U(E-tg)&=T\bigl(E-(E-tg)\bigr)
&&\text{(\eqref{eq:normalized-directional-gaps}, \(L=E-tg\))}\\
&=T(tg)
&&\text{(since \(E-(E-tg)=tg\))}\\
&=tT(g)
&&\text{(linearity of \(T\)).}
\end{aligned}
\label{eq:comparison-primal-gap}
\end{equation}
Together with \(T(g)>0\) and the choice of \(\rho\), this gives
\(0<d_U(E-tg)=tT(g)\le\rho T(g)<s_q\).
The curves \(E-tg\) and \(R_{tT(g)}\) have the same primal gap,
so their dual gaps agree for every \(0<t\le\rho\).
Expanding these two values gives two expressions for the same scalar
quantity throughout this interval.  

For every \(0<t\le\rho\), recall that
\(F_g(t,\cdot)=(E-tg)^{-1}\) and
\(d_U(E-tg)=tT(g)\in(0,s_q)\).
Thus both the transfer identity and the reference expansion apply:
\begin{equation}
\begin{aligned}
&t m(g)+t^2m(gg')+r_g(t)\\
&\quad=m(F_g(t,\cdot)-E)
&&\text{(scalar expansion \eqref{eq:scalar-inverse-expansion})}\\
&\quad=d_W(E-tg)
&&\text{(dual gap \eqref{eq:normalized-directional-gaps})}\\
&\quad=\Phi(d_U(E-tg))
&&\text{(transfer identity \eqref{eq:all-curve-transfer})}\\
&\quad=\Phi(tT(g))
&&\text{(primal gap \eqref{eq:comparison-primal-gap})}\\
&\quad=tT(g)+\frac{c}{2}(tT(g))^2+r_q(tT(g))
&&\text{(\eqref{eq:transfer-second-order}, \(s=tT(g)\))}\\
&\quad=tT(g)+\frac{c}{2}t^2T(g)^2+r_q(tT(g))
&&\text{(expand the square).}
\end{aligned}
\label{eq:comparison-of-variations}
\end{equation}
We now compare the coefficients of \(t\) and \(t^2\) in
\eqref{eq:comparison-of-variations}. For the fixed direction
\(g\), write
\[
\alpha_g:=m(g)-T(g),
\qquad
\beta_g:=m(gg')-\frac{c}{2}T(g)^2.
\]
These are the differences between the linear and quadratic
coefficients, respectively. Both are fixed real numbers independent
of \(t\); our next task is to prove that both are zero.
Rearranging the exact identity gives
\[
t\alpha_g+t^2\beta_g=r_q(tT(g))-r_g(t)
\qquad(0<t\le\rho).
\]

To compare the coefficients, we need one remainder bound valid
for every \(0<t\le\rho\).
Recall that \(q\) was chosen with \(\|q'\|_\infty=1\), and that
\(T,m\) were normalized by \(T(q)=m(q)=1\).
These choices have already been made and are kept throughout the
proof. For the present comparison, \(g\) and \(\rho\) have also
been fixed. Thus \(C_q\), \(T(g)>0\), and \(c=2m(qq')\) are fixed
numbers; only \(t\) will vary.

We first check that the two whole parameter intervals admit
uniform remainder bounds. Since \(g'\) and \(q'\) are decreasing,
for every \(x\in[0,1]\) we have
\[
\begin{aligned}
1-u g'(x)
&\ge1-u g'(0)\ge1-\rho g'(0)>0
&&(0\le u\le\rho),\\
1-s q'(x)
&\ge1-s q'(0)\ge1-\rho T(g)q'(0)>0
&&(0\le s\le\rho T(g)).
\end{aligned}
\]
The final strict inequalities follow from
\eqref{eq:comparison-radius}. By the definitions of \(I_g,I_q\),
\[
[0,\rho]\subset I_g,\qquad [0,\rho T(g)]\subset I_q.
\]
Thus \eqref{eq:inverse-remainder-constant} gives finite constants
\(C(g,[0,\rho])\) and \(C(q,[0,\rho T(g)])\).
Each is defined by taking a maximum over its entire fixed
parameter interval and the spatial interval \([0,1]\).

Now let \(0<t\le\rho\). Since \(T(g)>0\), the reference parameter
satisfies \(0<tT(g)\le\rho T(g)\).
Applying the scalar remainder bound from part~(ii) first at
\(t\), then at \(s=tT(g)\), gives
\[
\begin{aligned}
|r_g(t)|
&\le C_q C(g,[0,\rho])t^3,\\
|r_q(tT(g))|
&\le C_q C(q,[0,\rho T(g)])\,(tT(g))^3\\
&=C_q T(g)^3C(q,[0,\rho T(g)])\,t^3.
\end{aligned}
\]
In particular, the factor \(T(g)^3\) comes from the change of
parameter \(s=tT(g)\). Define
\[
C_{\mathrm{comp}}(g,\rho)
:=C_q\bigl(C(g,[0,\rho])+T(g)^3C(q,[0,\rho T(g)])\bigr)<\infty.
\]
Both intervals in this definition remain fixed as \(t\) varies,
so \(C_{\mathrm{comp}}(g,\rho)\) is independent of \(t\).

Combining the exact coefficient identity with these bounds yields
\[
\begin{aligned}
|t\alpha_g+t^2\beta_g|
&=|r_q(tT(g))-r_g(t)|\\
&\le |r_q(tT(g))|+|r_g(t)|\\
&\le C_q C(q,[0,\rho T(g)])T(g)^3t^3
+C_q C(g,[0,\rho])t^3\\
&=t^3C_{\mathrm{comp}}(g,\rho).
\end{aligned}
\]
Since \(t>0\), dividing this inequality by \(t\) gives
\begin{equation}
|\alpha_g+t\beta_g|
=\frac{|t\alpha_g+t^2\beta_g|}{t}
\le t^2C_{\mathrm{comp}}(g,\rho)
\qquad(0<t\le\rho).
\label{eq:first-variation-comparison}
\end{equation}
Given the \(g,\rho\) are fixed, the same finite constant \(\alpha_g,\beta_g,C_{\mathrm{comp}}(g,\rho)\) therefore holds for every
admissible \(t\). 

\noindent\emph{(iv) Equality of the linear functionals.}
Suppose, for a contradiction, that
\(\delta:=|\alpha_g|=|m(g)-T(g)|>0\).
Because \(\beta_g\) and \(C_{\mathrm{comp}}(g,\rho)\) are fixed terms independent of \(t\),
we can choose a single \(t\in(0,\rho]\) small enough that
\[
t|\beta_g|+t^2C_{\mathrm{comp}}(g,\rho)<\delta.
\]
The comparison inequality holds at this chosen \(t\). Applying
it and the triangle inequality gives
\[
\begin{aligned}
\delta=|\alpha_g|
&\le|\alpha_g+t\beta_g|+t|\beta_g|\\
&\le t^2C_{\mathrm{comp}}(g,\rho)+t|\beta_g|\\
&<\delta,
\end{aligned}
\]
which contradicts \(\delta>0\).
Thus \(\alpha_g=0\), which by its definition means \(m(g)=T(g)\).
Since the initial choice of \(g\in V_{\mathrm{conc}}\) was arbitrary,
we have proved
\begin{equation}
m(g)=T(g)\qquad(g\in V_{\mathrm{conc}}).
\label{eq:first-coefficient-on-cone}
\end{equation}

We will now extend the linear-part equality relationship, currently restricted to the concave directions of \(g\), to all directions. For an arbitrary fixed
\(h\in V\), both \(q\) and \(q+sh\) retain strict concavity
when \(|s|\) is small. Subtracting their identities recovers the
linear equality in the direction \(h\).

For \(h\in V\), let
\[
\theta_h:=\frac{h''}{-q''},\qquad
M_h^\pm:=\max_{x\in[0,1]}\max\{\pm\theta_h(x),0\}.
\]
With the convention \(1/0:=+\infty\), define
\[
J_h:=\left(-\frac1{M_h^-},\frac1{M_h^+}\right)
=\{s\in\mathbb R:\min_{[0,1]}(1-s\theta_h)>0\}.
\]
For \(s\ge0\), strict positivity of \(1-s\theta_h\) everywhere
is equivalent to \(sM_h^+<1\); for \(s\le0\), it is equivalent to
\((-s)M_h^-<1\). These are precisely the right and left
bounds in \(J_h\); when one of \(M_h^\pm\) is zero, that
bound imposes no restriction.
For \(s\in J_h\),
\[
q+sh\in V,\qquad
-(q+sh)''=(-q'')(1-s\theta_h)
\ge a_q\min_{[0,1]}(1-s\theta_h)>0.
\]
The interval \(J_h\) contains nonzero values of \(s\),
because \(a_q>0\) and \(\|h''\|_\infty<\infty\).
Choose one such \(s\). Both \(q\) and \(q+sh\) belong to
\(V_{\mathrm{conc}}\), so \eqref{eq:first-coefficient-on-cone} gives
\[
m(q+sh)=T(q+sh),\qquad m(q)=T(q).
\]
Linearity of \(m\) and \(T\) rewrites the first equality as
\[
m(q)+s\,m(h)=T(q)+s\,T(h).
\]
Subtracting the equality at \(q\) yields
\(s(m(h)-T(h))=0\). Since the chosen \(s\) is nonzero, \(m(h)=T(h)\).
This proves \(m=T\) on all of \(V\), as \(h\) was arbitrary.
In particular, \(qq'\in V\), so the earlier definition of
\(c\) now also gives
\[
c=2m(qq')=2T(qq').
\]

\noindent\emph{(v) Equality of the quadratic coefficients.}
Return to a fixed \(g\in V_{\mathrm{conc}}\) and its chosen radius \(\rho\).
We have already proved \(\alpha_g=0\). Substituting this value into
\eqref{eq:first-variation-comparison}, and dividing by \(t>0\), gives
\[
|t\beta_g|\le t^2C_{\mathrm{comp}}(g,\rho),
\qquad
|\beta_g|\le tC_{\mathrm{comp}}(g,\rho)
\qquad(0<t\le\rho).
\]
Again this bound holds for every \(t\) in the interval, with the
same \(\beta_g\) and \(C_{\mathrm{comp}}(g,\rho)\).
If \(\eta:=|\beta_g|>0\), choose \(t\in(0,\rho]\) so small that
\(tC_{\mathrm{comp}}(g,\rho)<\eta\). The bound at that \(t\) would give
\[
\eta=|\beta_g|\le tC_{\mathrm{comp}}(g,\rho)<\eta,
\]
a contradiction. Hence \(\beta_g=0\), and its definition gives
\[
m(gg')=\frac{c}{2}T(g)^2.
\]
To express this identity entirely in terms of \(T\), observe that
\(gg'\in V\), since \(g(0)=g(1)=0\).
The equality \(m=T\) on \(V\), proved in part~(iv),
therefore gives
\begin{equation}
T(gg')=\frac{c}{2}T(g)^2
\qquad(g\in V_{\mathrm{conc}}).
\label{eq:quadratic-on-cone}
\end{equation}

\noindent\emph{(vi) Symmetric finite difference.}
It remains to extend the quadratic identity from
\(V_{\mathrm{conc}}\) to all of \(V\). For an arbitrary direction
\(h\), the two perturbations \(q+sh\) and \(q-sh\) remain
uniformly concave when \(s>0\) is small. Adding their identities
will cancel the terms involving both \(q\) and \(h\), leaving
the quadratic expression in \(h\) alone.

For \(v\in V\), define
\[
Z(v):=T(vv')-\frac{c}{2}T(v)^2.
\]
The product \(vv'\) belongs to \(V\), since \(v\) is smooth
and vanishes at both endpoints. Our goal is to prove \(Z(h)=0\).
Equation~\eqref{eq:quadratic-on-cone} already gives \(Z(v)=0\)
for every \(v\in V_{\mathrm{conc}}\), including \(v=q\).

Fix \(h\in V\), without assuming concavity, and choose
\(s>0\) with \(s\|\theta_h\|_\infty<1\), so that \(s,-s\in J_h\).
Both perturbations vanish at the endpoints, and
\[
-(q\pm sh)''=(-q'')(1\mp s\theta_h)
\ge a_q(1-s\|\theta_h\|_\infty)>0.
\]
Thus \(q+sh,q-sh\in V_{\mathrm{conc}}\), and
\(Z(q+sh)=Z(q-sh)=Z(q)=0\).

The reason for using both signs is visible in the expansions.
Since \(s\) is a scalar and \(T(q)=1\),
\[
\begin{aligned}
(q\pm sh)(q\pm sh)'
&=qq'\pm s(qh'+hq')+s^2hh',\\
T(q\pm sh)^2&=1\pm2sT(h)+s^2T(h)^2.
\end{aligned}
\]
Although \(T\) is linear, \(Z\) contains the products \(vv'\)
and \(T(v)^2\). We therefore expand these products before
collecting terms by their powers of \(s\):
\[
\begin{aligned}
Z(q+sh)
&=T(qq')+sT(qh'+hq')+s^2T(hh') -\frac c2-csT(h)-\frac c2s^2T(h)^2\\
&=\underbrace{\left(T(qq')-\frac c2(T(q)^2\right)}_{Z(q)}+s\bigl(T(qh'+hq')-cT(h)\bigr)+s^2\underbrace{\left(T(hh')-\frac c2T(h)^2\right)}_{Z(h)}\\
&= Z(q)+ s\bigl(T(qh'+hq')-cT(h)\bigr)+s^2Z(h)
\end{aligned}
\]
Note the first bracket is \(Z(q)\) because \(T(q)=1\) so $\frac c2=\frac c2(T(q)^2$, and
the coefficient of \(s^2\) is exactly the defining expression
for \(Z(h)\). Replacing \(s\) by \(-s\) gives
\[
Z(q-sh)
=Z(q)-s\bigl(T(qh'+hq')-cT(h)\bigr)+s^2Z(h).
\]
The terms linear in \(s\) have opposite signs, while the
quadratic terms agree. Adding the two identities and subtracting
\(2Z(q)\) therefore yields
\[
0=Z(q+sh)+Z(q-sh)-2Z(q)=2s^2Z(h).
\]
This is the symmetric second finite difference: it isolates
the quadratic term by canceling the constant and linear terms.
The calculation is exact at the chosen nonzero \(s\), so
division by \(2s^2>0\) gives \(Z(h)=0\).
Since \(h\in V\) was arbitrary, we have proved
\(T(hh')=(c/2)T(h)^2\) for every \(h\in V\).

\noindent\emph{(vii) Polarization.}
We can now apply the quadratic identity to \(h+k\) and subtract
its values at \(h\) and \(k\). The mixed terms that remain are
\(hk'+kh'=(hk)'\), producing the stated bilinear identity.

For \(h,k\in V\) and \(p\in\{0,1\}\),
\[
(hk)'=h'k+hk',\qquad
(hk)'(p)=h'(p)k(p)+h(p)k'(p)=0.
\]
Thus \((hk)'\in V\). Since \(h+k\in V\) as well, the quadratic
identity from part~(vi) applies to each of \(h+k,h,k\).
Using these three identities and the linearity of \(T\), we obtain
\[
\begin{aligned}
T((hk)')
&=T(hk'+kh')\\
&=T((h+k)(h'+k'))-T(hh')-T(kk')\\
&=\frac{c}{2}\bigl(T(h+k)^2-T(h)^2-T(k)^2\bigr)\\
&=\frac{c}{2}\bigl((T(h)+T(k))^2
-T(h)^2-T(k)^2\bigr)\\
&=c T(h)T(k).
\end{aligned}
\]
All representative evaluations in this proof occur on
\(R_s\), \(0\le s<s_q\), or \(E-tg\),
\(g\in V_{\mathrm{conc}}\), \(0\le t\le\rho\) with
\eqref{eq:comparison-radius}.
\end{proof}

\section{Classification and Lorenz area representation}
\label{sec:proof-classification}

\subsection{Classification of the positive linear functional}
\label{subsec:classification}

The bilinear identity reduces the problem to a classification
of one positive linear functional. We will obtain three possibilities:
normalized area and two endpoint-slope functionals. This step uses
the identity algebraically; continuity of the original representative
will distinguish the alternatives in the next subsection.

The relevant features of \(q\) that will be used are its positivity
in the interior and its simple zeros at the endpoints. The curvature
bounds were needed to obtain admissible comparisons.

\begin{lemma}[Area and the two endpoint alternatives]
\label{lem:functional-classification}
Let \(T:V\to\mathbb R\) be positive and linear,
\(T(q)=1\), and let \(c\in\mathbb R\) satisfy
\eqref{eq:bilinear-compatibility}. Exactly one of the following holds:
\[
\begin{array}{c|c}
c&T(h)\quad(h\in V)\\ \hline
0&A(h)/A(q)\\[3pt]
2q'(0)&h'(0)/q'(0)\\[3pt]
2q'(1)&h'(1)/q'(1).
\end{array}
\]
Each displayed pair satisfies all the hypotheses.
\end{lemma}
\begin{proof}
Parts~(i)--(iv) determine the possible forms of a solution.
Part~(v) then checks that all three candidates satisfy the
hypotheses and that they give distinct functionals.

\noindent\emph{(i) The zero-coefficient case.}
When \(c=0\), the bilinear identity makes \(T\) vanish on every
derivative \((qk)'\) with \(k\in V\). Every zero-area
direction can be written in this form: its primitive vanishes
together with its first derivative at both endpoints. It follows
that area is the only information that \(T\) can retain.

Suppose \(c=0\). For \(v\in V\) with \(A(v)=0\), set
\[
H(x):=\int_0^xv(s)\,ds.
\]
Since \(v\) is smooth, \(H\in C^\infty([0,1])\) and \(H'=v\).
Its endpoint values and derivatives satisfy
\[
\begin{aligned}
H(0)&=0,& H(1)&=A(v)=0,\\
H'(0)&=v(0)=0,& H'(1)&=v(1)=0.
\end{aligned}
\]
In particular, \(H\in V\), so division by \(q\) is covered by
the earlier smooth-division result.
Lemma~\ref{lem:quadratic-division} gives
\[
k:=H/q\in C^\infty([0,1]),\qquad
k(0)=\frac{H'(0)}{q'(0)}=0,\quad
k(1)=\frac{H'(1)}{q'(1)}=0.
\]
Hence
\[
k\in V,\qquad qk=H,\qquad(qk)'=v,\qquad
T(v)=T((qk)')=c T(q)T(k)=0.
\]
For a general direction, subtracting a multiple of \(q\) with
the same area leaves a zero-area direction. The preceding
calculation then determines its value under \(T\).
For arbitrary \(h\in V\),
\[
\begin{aligned}
v_h&:=h-\frac{A(h)}{A(q)}q\in V,\\
A(v_h)&=A(h)-\frac{A(h)}{A(q)}A(q)=0,\\
0=T(v_h)&=T(h)-\frac{A(h)}{A(q)}T(q),\\
T(h)&=\frac{A(h)}{A(q)}.
\end{aligned}
\]

\noindent\emph{(ii) A multiplicative functional when \(c\ne0\).}
Multiplication by \(q\) turns any smooth function into an
endpoint-vanishing direction. This lets us work on the full algebra
\(C^\infty([0,1])\), which contains the constant function \(1\).
The key calculation applies the bilinear identity to the same product
\(q^2fg\) in two ways: as \((qf)(qg)\) and as \(q(qfg)\).

Suppose \(c\ne0\), and define
\[
S:C^\infty([0,1])\to\mathbb R,\qquad S(f):=T(qf).
\]
For \(a,b\in\mathbb R\), \(f,g\in C^\infty([0,1])\),
the products \(qf,qg,qfg\) belong to \(V\).
Linearity and normalization of \(T\) give
\[
S(af+bg)=aS(f)+bS(g),\qquad S(1)=T(q)=1.
\]
If \(f\ge0\), then \(qf\ge0\), so positivity of \(T\) gives
\(S(f)=T(qf)\ge0\).
The products \((qf)(qg)\) and \(q(qfg)\) are the same function.
Applying the bilinear identity to these two pairs gives two expressions
for the same scalar:
\[
\begin{aligned}
T((q^2fg)')
&=cT(qf)T(qg)=cS(f)S(g),\\
T((q^2fg)')
&=cT(q)T(qfg)=cS(fg).
\end{aligned}
\]
Division by \(c\ne0\) proves
\begin{equation}
S(fg)=S(f)S(g)\qquad(f,g\in C^\infty([0,1])).
\label{eq:multiplicative-functional}
\end{equation}

\noindent\emph{(iii) Evaluation at one point.}
A positive normalized multiplicative functional on this algebra must evaluate
functions at a single point.
Its value on the coordinate function
\(E(x)=x\) identifies that point. The factorization of \(f-f(p)\)
below gives a direct proof for every smooth \(f\).

Let \(p:=S(E)\), where \(E(x)=x\). Positivity gives
\[
p=S(E)\ge0,\qquad 1-p=S(1-E)\ge0,
\qquad p\in[0,1].
\]
For \(f\in C^\infty([0,1])\), define
\[
r_{f,p}(x):=\int_0^1f'(p+s(x-p))\,ds.
\]
For \(x\in[0,1]\), \(s\in[0,1]\),
\[
p+s(x-p)=(1-s)p+sx\in[0,1].
\]
The fundamental theorem of calculus and
Lemma~\ref{lem:parameter-integration} give
\[
\begin{aligned}
f(x)-f(p)&=(x-p)r_{f,p}(x),\\
r_{f,p}^{(j)}(x)&=\int_0^1s^jf^{(j+1)}(p+s(x-p))\,ds
\qquad(j\ge0),\\
r_{f,p}&\in C^\infty([0,1]).
\end{aligned}
\]
The identity includes \(x=p\) and \(p\in\{0,1\}\).
By the choice of \(p\), \(S(E-p\,1)=0\); multiplicativity
therefore makes the entire remainder \((E-p\,1)r_{f,p}\)
vanish under \(S\).
Thus
\[
\begin{aligned}
S(f)&=S(f(p)\,1+(E-p\,1)r_{f,p})\\
&=f(p)S(1)+S(E-p\,1)S(r_{f,p})\\
&=f(p)+(S(E)-pS(1))S(r_{f,p})\\
&=f(p).
\end{aligned}
\]
Therefore
\begin{equation}
S(f)=f(p)\qquad(f\in C^\infty([0,1])).
\label{eq:point-evaluation}
\end{equation}

\noindent\emph{(iv) The point must be an endpoint.}
Multiplicativity locates a point, but the derivative in the original
bilinear identity further restricts its position. Choosing
a function whose value at \(p\) is zero and whose derivative there
is one will isolate \(q(p)\), forcing \(p\) to an endpoint.

Apply \eqref{eq:bilinear-compatibility} to \(q,qf\):
\[
\begin{aligned}
(q^2f)'&=2qq'f+q^2f'=q(2q'f+qf'),\\
S(2q'f+qf')&=T((q^2f)')
=c T(q)T(qf)=c S(f).
\end{aligned}
\]
Together with \eqref{eq:point-evaluation}, this yields
\begin{equation}
2q'(p)f(p)+q(p)f'(p)=c f(p)
\qquad(f\in C^\infty([0,1])).
\label{eq:evaluation-derivative-constraint}
\end{equation}
Taking \(f=E-p\,1\) gives \(f(p)=0\) and \(f'(p)=1\), so the
preceding identity implies \(q(p)=0\).
Since \(q>0\) on \((0,1)\) and \(p\in[0,1]\), this forces
\(p\in\{0,1\}\).
Taking \(f=1\) in the same identity then gives
\[
c=2q'(p)=
\begin{cases}
2q'(0)>0,&p=0,\\
2q'(1)<0,&p=1.
\end{cases}
\]
For the final recovery of \(T\), the quotient \(h/q\) is
understood as its smooth extension to both endpoints.
Since \(q\) has simple endpoint zeros, the extended quotient
has value \(h'(p)/q'(p)\) at an endpoint \(p\).
Finally, for \(h\in V\),
\[
T(h)=T(q(h/q))=S(h/q)=(h/q)(p)
=\begin{cases}
h'(0)/q'(0),&p=0,\\
h'(1)/q'(1),&p=1,
\end{cases}
\]
by Lemma~\ref{lem:quadratic-division}.

\noindent\emph{(v) Verification and exclusivity.}
So far, every solution has been shown to have one of the three
forms in the table. We now check the converse: each candidate
must be positive, linear, normalized at \(q\), and satisfy
the bilinear identity with its listed coefficient \(c\).
Define
\[
T_A(h):=\frac{A(h)}{A(q)},\qquad
T_{\rm left}(h):=\frac{h'(0)}{q'(0)},\qquad
T_{\rm right}(h):=\frac{h'(1)}{q'(1)}.
\]
\smallskip
\noindent\emph{Positivity and normalization.}
All three are linear. For \(a,b\in\mathbb R\) and \(h,k\in V\),
\(ah+bk\in V\), and
\[
\begin{aligned}
T_A(ah+bk)
&=\frac{aA(h)+bA(k)}{A(q)}
=aT_A(h)+bT_A(k),\\
T_{\rm left}(ah+bk)
&=\frac{ah'(0)+bk'(0)}{q'(0)}
=aT_{\rm left}(h)+bT_{\rm left}(k),\\
T_{\rm right}(ah+bk)
&=\frac{ah'(1)+bk'(1)}{q'(1)}
=aT_{\rm right}(h)+bT_{\rm right}(k).
\end{aligned}
\]
If \(h\in V\), \(h\ge0\), then
\[
\begin{aligned}
A(h)&=\int_0^1h(x)\,dx\ge0,\\
h'(0)&=\lim_{x\downarrow0}\frac{h(x)-h(0)}x
=\lim_{x\downarrow0}\frac{h(x)}x\ge0,\\
h'(1)&=\lim_{x\uparrow1}\frac{h(x)-h(1)}{x-1}
=-\lim_{x\uparrow1}\frac{h(x)}{1-x}\le0.
\end{aligned}
\]
At the right endpoint, \(h'(1)\le0\) is divided by
\(q'(1)<0\), so this candidate also takes nonnegative values
on nonnegative directions.
Together with \(A(q)>0\) and \(q'(0)>0>q'(1)\), these inequalities
prove positivity of all three functionals:
\[
T_A(h)=\frac{A(h)}{A(q)}\ge0,\qquad
T_{\rm left}(h)=\frac{h'(0)}{q'(0)}\ge0,\qquad
T_{\rm right}(h)=\frac{-h'(1)}{-q'(1)}\ge0.
\]
Their normalizations are
\[
T_A(q)=\frac{A(q)}{A(q)}=1,\qquad
T_{\rm left}(q)=\frac{q'(0)}{q'(0)}=1,\qquad
T_{\rm right}(q)=\frac{q'(1)}{q'(1)}=1.
\]
\smallskip
\noindent\emph{The bilinear identity for the area and endpoint candidates.}
For \(h,k\in V\), the input \((hk)'\) belongs to \(V\):
at either endpoint \(p\), both \(h(p)\) and \(k(p)\) vanish,
so
\[
(hk)'(p)=h'(p)k(p)+h(p)k'(p)=0
\qquad(p\in\{0,1\}).
\]
Since \(h,k\) are smooth, \((hk)'\in C^\infty([0,1])\).
Together with \((hk)'(0)=(hk)'(1)=0\), this proves \((hk)'\in V\).
For the area candidate,
\[
\begin{aligned}
T_A((hk)')
&=\frac{\int_0^1(hk)'(x)\,dx}{A(q)}\\
&=\frac{h(1)k(1)-h(0)k(0)}{A(q)}
=0=0\cdot T_A(h)T_A(k).
\end{aligned}
\]
This verifies the bilinear identity with \(c=0\).

An endpoint candidate takes the derivative of its input.
Here that input is already \((hk)'\), so the calculation
requires the second derivative of the product:
\[
\begin{aligned}
(hk)''&=h''k+2h'k'+hk'',\\
(hk)''(0)&=2h'(0)k'(0),\qquad
(hk)''(1)=2h'(1)k'(1).
\end{aligned}
\]
Given that
\(h(0)=h(1)=k(0)=k(1)=0\) by definition,
the endpoint conditions \(h(p)=k(p)=0\) remove the two product
terms, whereas \(T(q)=1\) fixes the scale of the functional.
Indeed, for an endpoint candidate \(T(h)=a h'(p)\), where
\(p\in\{0,1\}\) and \(a\in\mathbb R\),
\[
q'(p)\ne0,\qquad
1=T(q)=a q'(p)
\quad\Longrightarrow\quad
a=\frac1{q'(p)}.
\]
Hence \(T(h)=h'(p)/q'(p)\). This is the existing normalization
at \(q\); no additional endpoint normalization is imposed.
Combining the endpoint conditions with this normalized scale gives
\[
\begin{aligned}
T_{\rm left}((hk)')
&=\frac{(hk)''(0)}{q'(0)}
=\frac{(h''k)(0)+2h'(0)k'(0)+(hk'')(0)}{q'(0)}\\
&=2q'(0)\frac{h'(0)}{q'(0)}\frac{k'(0)}{q'(0)}
=2q'(0)T_{\rm left}(h)T_{\rm left}(k),\\[3pt]
T_{\rm right}((hk)')
&=\frac{(hk)''(1)}{q'(1)}
=\frac{(h''k)(1)+2h'(1)k'(1)+(hk'')(1)}{q'(1)}\\
&=2q'(1)\frac{h'(1)}{q'(1)}\frac{k'(1)}{q'(1)}
=2q'(1)T_{\rm right}(h)T_{\rm right}(k).
\end{aligned}
\]
\smallskip
\noindent\emph{Why the three candidates are distinct.}
For any normalized solution, the pair \(h=k=q\) determines
\(c=T((q^2)')=2T(qq')\).
Indeed, normalization makes the right-hand side of the bilinear identity
equal to \(cT(q)^2=c\). Explicitly,
\[
c=cT(q)^2=T((q^2)')=T(2qq')=2T(qq').
\]
Thus \(c\) is determined by \(T\) itself:
if two candidates were the same functional, they would necessarily
have the same coefficient.
For the same input \((q^2)'\in V\), the candidate values are
\[
\begin{aligned}
T_A((q^2)')
&=\frac{q(1)^2-q(0)^2}{A(q)}=0,\\
T_{\rm left}((q^2)')
&=\frac{(q^2)''(0)}{q'(0)}
=\frac{2(q'(0))^2}{q'(0)}=2q'(0)>0,\\
T_{\rm right}((q^2)')
&=\frac{(q^2)''(1)}{q'(1)}
=\frac{2(q'(1))^2}{q'(1)}=2q'(1)<0.
\end{aligned}
\]
Hence
\[
T_{\rm right}((q^2)')<T_A((q^2)')<T_{\rm left}((q^2)').
\]
In particular,
\[
T_A\ne T_{\rm left},\qquad T_A\ne T_{\rm right},\qquad
T_{\rm left}\ne T_{\rm right}.
\]
Since \(2q'(0)>0>2q'(1)\), the three values
\(0,2q'(0),2q'(1)\) are distinct, proving exclusivity.
The two endpoint candidates are valid solutions of this identity;
the next subsection excludes them using continuity of the original
representative.
\end{proof}

\subsection{Boundary layers, endpoint exclusion, and the area formula}
\label{subsec:area-identification}

Endpoint derivatives are sensitive to changes that can be very
small in the uniform norm. To exclude them within the actual
Lorenz domain, we construct smooth curves whose distance from
\(E\) tends to zero while a fixed endpoint-slope deviation persists.
Either endpoint functional would assign these curves a fixed
positive gap, contradicting continuity at \(E\).

\begin{proposition}[Continuity excludes both endpoint alternatives]
\label{prop:endpoint-exclusion}
Neither \(T(h)=h'(0)/q'(0)\) nor \(T(h)=h'(1)/q'(1)\) is compatible
with \eqref{eq:representative-continuity} at \(E\).
\end{proposition}
\begin{proof}
Fix an endpoint-slope deviation \(\eta\in(0,1)\).
The construction will concentrate the curvature of a gap near
one endpoint. This allows the gap to have arbitrarily small
height while its slope at that endpoint stays fixed.

\smallskip
\noindent\emph{(i) Constructing a gap from its curvature.}
Choose any nonnegative \(\omega\in C_c^\infty((0,1))\) with
\(\int_0^1\omega(u)\,du=1\), and extend it by zero outside
\((0,1)\). Such a profile is obtained by normalizing a nonzero
nonnegative smooth bump; its particular shape will not matter.
Compact support in \((0,1)\) means that \(\omega\) vanishes
near both endpoints, so the extension by zero is smooth.
Put \(\mu:=\int_0^1u\omega(u)\,du\in(0,1)\).
For \(0<r<1\), put
\[
\omega_r(x):=\frac1r\omega(x/r),\qquad
\mu_r:=\int_0^1x\omega_r(x)\,dx.
\]
The parameter \(r\) is the width of the region carrying curvature.
The number \(\mu_r\) is the first moment of \(\omega_r\); it determines
the normalization needed for the endpoint condition at \(1\).
A change of variables gives
\[
\supp\omega_r\subset(0,r),\qquad
\int_0^1\omega_r(x)\,dx=1,\qquad
\mu_r=r\mu\in(0,r).
\]

We seek a concave gap \(b_r\) with initial slope \(1\) and
zero values at both endpoints. The positive scale \(a_r\)
will be chosen to meet the condition at \(1\).
Prescribing its curvature by
\[
b_r''=-a_r\omega_r,\qquad b_r(0)=0,\qquad b_r'(0)=1
\]
gives, by two integrations,
\[
\begin{aligned}
b_r'(x)&=1-a_r\int_0^x\omega_r(u)\,du,\\
b_r(x)&=\int_0^x b_r'(v)\,dv
=x-a_r\int_0^x\int_0^v\omega_r(u)\,du\,dv\\
&=x-a_r\int_0^x(x-u)\omega_r(u)\,du
\qquad\text{(reverse the integration order)},\\
b_r(1)&=1-a_r\left(
\int_0^1\omega_r(u)\,du-\int_0^1u\omega_r(u)\,du
\right)\\
&=1-a_r(1-\mu_r).
\end{aligned}
\]
Thus the remaining endpoint condition \(b_r(1)=0\) determines
the normalization \(a_r=1/(1-\mu_r)>0\).
In particular, \(b_r\) is smooth, vanishes at both endpoints,
and
\[
b_r'(x)=1-\frac{\int_0^x\omega_r(u)\,du}{1-\mu_r},
\qquad
b_r''(x)=-\frac{\omega_r(x)}{1-\mu_r}\le0.
\]
Recall that \(V=\{h\in C^\infty([0,1]):h(0)=h(1)=0\}\).
The smoothness and endpoint conditions just proved give \(b_r\in V\).
Concavity and the endpoint values give \(b_r\ge0\).
Moreover, \(0\le\int_0^x\omega_r\le1\) and
\(\mu_r<r<1\) imply
\[
-\frac{\mu_r}{1-\mu_r}
\le b_r'(x)\le1,\qquad
b_r'(0)=1,\quad b_r'(1)=-\frac{\mu_r}{1-\mu_r}.
\]

The two endpoint slopes control the height through the first moment:
\[
\begin{aligned}
b_r(x)&=\int_0^x b_r'(u)\,du\le x,\\
b_r(x)&=-\int_x^1 b_r'(u)\,du
\le\frac{\mu_r}{1-\mu_r}(1-x).
\end{aligned}
\]
These affine upper bounds meet at \(x=\mu_r\), and hence
\[
0\le b_r(x)\le
\min\left\{x,\frac{\mu_r}{1-\mu_r}(1-x)\right\}
\le\mu_r\qquad(0\le x\le1).
\]
We have therefore obtained the two properties needed below:
\[
\|b_r\|_\infty\le\mu_r=\mu r,\qquad b_r'(0)=1.
\]
The width controls the height, independently of the fixed
endpoint slope.

\smallskip
\noindent\emph{(ii) Admissible curves for both endpoints.}
Recall that \(E(x)=x\), and that \(\mathcal{L}^\uparrow\)
consists of continuous, strictly increasing, convex curves
joining \((0,0)\) to \((1,1)\). For smooth curves, positivity
of the first derivative and nonnegativity of the second
derivative verify the required shape conditions.
Use this gap and its reflection to define
\[
L_r^{\mathrm{left}}(x):=x-\eta b_r(x),\qquad
L_r^{\mathrm{right}}(x):=x-\eta b_r(1-x).
\]
Reflection moves the fixed gap slope \(b_r'(0)=1\) to a
gap slope \(-1\) at the right endpoint.
Both curves fix \(0\) and \(1\). Since \(b_r''\le0\),
\(b_r'\) attains its maximum at \(0\) and its minimum at \(1\).
Their minimum slopes are therefore
\[
\begin{aligned}
\min_x(L_r^{\mathrm{left}})'(x)
&=\min_x\bigl(1-\eta b_r'(x)\bigr)
=1-\eta b_r'(0)=1-\eta,\\
\min_x(L_r^{\mathrm{right}})'(x)
&=\min_x\bigl(1+\eta b_r'(1-x)\bigr)\\
&=1+\eta b_r'(1)
=1-\frac{\eta\mu_r}{1-\mu_r}\\
&=\frac{1-(1+\eta)\mu r}{1-\mu r}.
\end{aligned}
\]
Thus both slopes are positive precisely when
\((1+\eta)\mu r<1\), since \(0<\eta<1\).
Set
\[
R_\eta:=\min\left\{1,\frac1{(1+\eta)\mu}\right\},\qquad
s_r:=\min\left\{1-\eta,
\frac{1-(1+\eta)\mu r}{1-\mu r}\right\}>0
\quad(0<r<R_\eta).
\]
For these widths, their slopes and curvatures satisfy
\[
\begin{aligned}
(L_r^{\mathrm{left}})'(x)
&=1-\eta b_r'(x)\ge1-\eta\ge s_r>0,\\
(L_r^{\mathrm{right}})'(x)
&=1+\eta b_r'(1-x)
\ge\frac{1-(1+\eta)\mu r}{1-\mu r}\ge s_r>0,\\
(L_r^{\mathrm{left}})''(x)
&=-\eta b_r''(x)\ge0,\\
(L_r^{\mathrm{right}})''(x)
&=-\eta b_r''(1-x)\ge0.
\end{aligned}
\]
For \(e\in\{\mathrm{left},\mathrm{right}\}\),
the fact that \(b_r\in V\) gives
\[
L_r^e\in C^\infty([0,1]),\qquad L_r^e(0)=0,\quad L_r^e(1)=1.
\]
The slope bound also gives strict increase:
\[
L_r^e(y)-L_r^e(x)
=\int_x^y(L_r^e)'(u)\,du
\ge s_r(y-x)>0
\qquad(0\le x<y\le1).
\]
For \(0\le x<y<z\le1\), \eqref{eq:curvature-shape}
applied to \(P=L_r^e\) gives
\[
\begin{aligned}
&\frac{z-y}{z-x}L_r^e(x)+\frac{y-x}{z-x}L_r^e(z)-L_r^e(y)\\
&\quad=\frac1{z-x}\int_x^y\int_y^z\int_s^t
(L_r^e)''(v)\,dv\,dt\,ds\ge0.
\end{aligned}
\]
Thus \(L_r^e\) is convex. Together with the endpoint values,
strict increase, and smoothness established above, this yields
\(L_r^e\in\mathcal{L}^\uparrow\cap C^\infty([0,1])\)
for \(e\in\{\mathrm{left},\mathrm{right}\}\).
The height estimate gives
\begin{equation}
\begin{aligned}
\|L_r^{\mathrm{left}}-E\|_\infty
&=\max_{0\le x\le1}|-\eta b_r(x)|
=\eta\|b_r\|_\infty\le\eta\mu r,\\
\|L_r^{\mathrm{right}}-E\|_\infty
&=\max_{0\le x\le1}|-\eta b_r(1-x)|\\
&\overset{u=1-x}{=}\eta\max_{0\le u\le1}|b_r(u)|
=\eta\|b_r\|_\infty\le\eta\mu r,
\end{aligned}
\label{eq:boundary-layer-uniform-limit}
\end{equation}
whereas the relevant endpoint slopes are exactly
\begin{equation}
(L_r^{\mathrm{left}})'(0)=1-\eta,\qquad
(L_r^{\mathrm{right}})'(1)=1+\eta.
\label{eq:fixed-slope-deviations}
\end{equation}
Thus \(r\) can make the uniform distance as small as required
without changing these slope deviations.

\smallskip
\noindent\emph{(iii) A direct contradiction with continuity.}
Recall that \(D_U(L):=U(E)-U(L)\),
\(\sigma_U:=T_0(q)>0\), and \(T=T_0/\sigma_U\).
For a smooth Lorenz curve, \(E-L\in V\); linearity and
\eqref{eq:original-primal-direction} therefore give
\[
\begin{aligned}
D_U(L)&=-T_0(L-E)=T_0(E-L)\\
&=\sigma_U T(E-L).
\end{aligned}
\]
Recall also the two candidate formulas from
Lemma~\ref{lem:functional-classification}:
\[
T_{\rm left}(h)=\frac{h'(0)}{q'(0)},\qquad
T_{\rm right}(h)=\frac{h'(1)}{q'(1)}
\qquad(h\in V),
\]
where \(q'(0)>0>q'(1)\). Here the input is \(h=E-L\),
so \(h'(p)=1-L'(p)\) at either endpoint \(p\).
If \(T=T_{\rm left}\), the primal gap of the left curve is
\[
\begin{aligned}
D_U(L_r^{\mathrm{left}})
&=\sigma_U T(E-L_r^{\mathrm{left}})\\
&=\sigma_U\frac{1-(L_r^{\mathrm{left}})'(0)}{q'(0)}
=\frac{\sigma_U\eta}{q'(0)}
=:G_{\mathrm{left}}>0.
\end{aligned}
\]
If \(T=T_{\rm right}\), the reflected curve instead gives
\[
\begin{aligned}
D_U(L_r^{\mathrm{right}})
&=\sigma_U T(E-L_r^{\mathrm{right}})\\
&=\sigma_U\frac{1-(L_r^{\mathrm{right}})'(1)}{q'(1)}
=\sigma_U\frac{1-(1+\eta)}{q'(1)}\\
&=\frac{-\sigma_U\eta}{q'(1)}
=\frac{\sigma_U\eta}{-q'(1)}
=:G_{\mathrm{right}}>0.
\end{aligned}
\]
Fix \(e\in\{\mathrm{left},\mathrm{right}\}\) and suppose \(T=T_e\).
The preceding formulas give, for every \(0<r<R_\eta\),
\[
L_r^e\in\mathcal{L},\qquad
\|L_r^e-E\|_\infty\le\eta\mu r.
\]
At the same time, its gap is independent of \(r\):
\[
D_U(L_r^e)=G_e>0,\qquad
G_e=\begin{cases}
\dfrac{\sigma_U\eta}{q'(0)},&e=\mathrm{left},\\[2pt]
\dfrac{\sigma_U\eta}{-q'(1)},&e=\mathrm{right}.
\end{cases}
\]
Fix \(0<\alpha\le1\) and apply continuity of \(U\) at \(E\)
with \(\varepsilon=\alpha G_e>0\). There is a \(\delta_e>0\)
such that, for every \(L\in\mathcal{L}\),
\[
\|L-E\|_\infty<\delta_e
\quad\Longrightarrow\quad
|U(L)-U(E)|<\alpha G_e.
\]
Choose any width
\[
0<r_e<\min\left\{R_\eta,\frac{\delta_e}{\eta\mu}\right\}.
\]
The interval is nonempty because \(R_\eta,\delta_e,\eta,\mu>0\).
The two restrictions respectively ensure admissibility and proximity to \(E\):
\[
\begin{aligned}
0<r_e<R_\eta
&\quad\Longrightarrow\quad L_{r_e}^{\,e}\in\mathcal{L},\\
r_e<\frac{\delta_e}{\eta\mu}
&\quad\Longrightarrow\quad
\|L_{r_e}^{\,e}-E\|_\infty
\le\eta\mu r_e<\delta_e.
\end{aligned}
\]
The endpoint formula and continuity therefore give
\[
\begin{aligned}
G_e&=D_U(L_{r_e}^{\,e})
=|U(L_{r_e}^{\,e})-U(E)|\\
&<\alpha G_e
\end{aligned}
\]
Contradiction, hence \(T\ne T_{\rm left}\) and \(T\ne T_{\rm right}\).
\end{proof}

Only normalized integration remains in the classification.
The directional formula therefore determines \(U\) on smooth Lorenz
curves. Shape-preserving approximation and continuity extend this
identity to the entire original domain.

\begin{proposition}[Area representation on the original domain]
\label{prop:all-curves-area}
For every \(L\in\mathcal{L}\),
\begin{equation}
U(L)=U(E)+\frac{\sigma_U}{A(q)}\bigl(A(L)-A(E)\bigr).
\label{eq:all-curves-area}
\end{equation}
\end{proposition}
\begin{proof}
Recall that \(A(f)=\int_0^1f(x)\,dx\), \(A(q)>0\), and
\(T_0=\sigma_U T\), with the same fixed scale \(\sigma_U>0\)
used throughout the proof.
Proposition~\ref{prop:second-order-compatibility},
Lemma~\ref{lem:functional-classification}, and
Proposition~\ref{prop:endpoint-exclusion} give
\[
c=0,\qquad T(h)=\frac{A(h)}{A(q)}
\qquad(h\in V).
\]
The condition \(T(q)=1\) fixes the area coefficient as \(1/A(q)\);
the relation \(T_0=\sigma_U T\) restores the scale of \(U\).
For \(L\in\mathcal{L}\cap C^\infty([0,1])\), the difference
\(L-E\) is smooth and vanishes at both endpoints, so it belongs
to \(V\). The identity \eqref{eq:original-primal-direction}
applies on this whole smooth domain and gives
\[
\begin{aligned}
L-E&\in V,\\
U(L)-U(E)&=T_0(L-E)=\sigma_U T(L-E)\\
&=\frac{\sigma_U}{A(q)}A(L-E)
=\frac{\sigma_U}{A(q)}\bigl(A(L)-A(E)\bigr).
\end{aligned}
\]
For arbitrary \(L\in\mathcal{L}\), put \(L_n=B_nL\).
Recall that \(B_nL\) is the Bernstein polynomial constructed
in Lemma~\ref{lem:bernstein-density}; it preserves the Lorenz
shape conditions. Each approximant therefore satisfies the
smooth area formula for the same representative \(U\) and
the same coefficient \(\sigma_U/A(q)\).
Lemma~\ref{lem:bernstein-density} yields
\[
L_n\in\mathcal{L}\cap C^\infty([0,1]),\qquad
\|L_n-L\|_\infty\longrightarrow0,
\]
and
\[
|A(L_n)-A(L)|
=\left|\int_0^1(L_n-L)\right|
\le\int_0^1|L_n-L|
\le\|L_n-L\|_\infty.
\]
Let
\[
R(L):=U(L)-U(E)-\frac{\sigma_U}{A(q)}\bigl(A(L)-A(E)\bigr).
\]
Since \(R(L_n)=0\),
\[
\begin{aligned}
|R(L)|&=|R(L)-R(L_n)|\\
&\le|U(L)-U(L_n)|
+\frac{\sigma_U}{A(q)}|A(L)-A(L_n)|\\
&\le|U(L)-U(L_n)|+\frac{\sigma_U}{A(q)}\|L-L_n\|_\infty.
\end{aligned}
\]
Continuity at \(L\) controls the first term through the local modulus
\[
\omega_{U,L}(r):=
\sup_{\substack{P\in\mathcal{L}\\\|P-L\|_\infty\le r}}
|U(P)-U(L)|.
\]
It is finite for sufficiently small \(r>0\), and
\(\omega_{U,L}(r)\to0\) as \(r\downarrow0\).
Given \(\varepsilon>0\), choose any \(r>0\) satisfying the total error budget
\[
\omega_{U,L}(r)+\frac{\sigma_U}{A(q)}r<\varepsilon.
\]
Such radii exist because both terms tend to zero.
Choose \(n\ge2\) with \(\|L_n-L\|_\infty\le r\). The preceding estimate gives
\[
|R(L)|
\le\omega_{U,L}(r)+\frac{\sigma_U}{A(q)}r<\varepsilon.
\]
Thus \((\forall\varepsilon>0)\ |R(L)|<\varepsilon\), so
\(R(L)=0\), which is \eqref{eq:all-curves-area}.
\end{proof}

\begin{proof}[Completion of the proof of Theorem~\ref{thm:direct-gini}]
Set
\[
c_1:=\frac{\sigma_U}{A(q)}>0,\qquad c_0:=U(E)-c_1A(E).
\]
For \(L_1,L_2\in\mathcal{L}\),
\[
\begin{aligned}
L_1\succsim L_2
&\Longleftrightarrow U(L_1)\ge U(L_2)\\
&\Longleftrightarrow c_0+c_1A(L_1)\ge c_0+c_1A(L_2)\\
&\Longleftrightarrow c_1\bigl(A(L_1)-A(L_2)\bigr)\ge0\\
&\Longleftrightarrow A(L_1)\ge A(L_2)\\
&\Longleftrightarrow 1-2A(L_1)\le1-2A(L_2)\\
&\Longleftrightarrow J(L_1)\le J(L_2).
\end{aligned}
\]
\end{proof}

\appendix
\section{Auxiliary domain and finite-dimensional facts}
\label{app:domain-facts}
These facts justify the domain and topology checks used in the text.
They do not assume the main theorem and do not impose any additional preference axiom.

\begin{lemma}[Domain geometry, density, and inversion]
\label{lem:domain-inversion}
The sets \(\mathcal{L}\), \(\mathcal{L}^\uparrow\), and \(Y\) are convex.
Every \(L\in\mathcal{L}\) satisfies \(0\le L(x)\le x\) on \([0,1]\).
The set \(\mathcal{L}^\uparrow\) is uniformly dense in \(\mathcal{L}\), and
\(\iota:\mathcal{L}^\uparrow\to Y\), \(\iota(L)=L^{-1}\), is a homeomorphism for the relative uniform topologies.
Under Axioms~\ref{ax:weak-order} and~\ref{ax:continuity}, nontriviality on \(\mathcal{L}\) implies nontriviality on \(\mathcal{L}^\uparrow\).
\end{lemma}

\begin{proof}
For \(L,M\in\mathcal{L}\) and \(\lambda\in[0,1]\), let \(N=\lambda L+(1-\lambda)M\).
It has the required endpoints and is continuous. For \(x\le y\),
\(N(y)-N(x)=\lambda(L(y)-L(x))+(1-\lambda)(M(y)-M(x))\ge0\).
For \(x,y,t\in[0,1]\), applying convexity to each summand gives
\[
\begin{aligned}
N(tx+(1-t)y)
&\le\lambda\bigl(tL(x)+(1-t)L(y)\bigr)\\
&\quad +(1-\lambda)\bigl(tM(x)+(1-t)M(y)\bigr)
=tN(x)+(1-t)N(y).
\end{aligned}
\]
When \(L,M\) are strictly increasing, the increment for \(x<y\) is strictly positive, including when \(\lambda=0\) or \(1\).
The same increment calculation and the reversed convexity inequality prove convexity of \(Y\).
Monotonicity and the endpoints give \(L(x)\ge0\), while convexity gives
\(L(x)\le(1-x)L(0)+xL(1)=x\).
For \(0<\varepsilon<1\), define \(L^\varepsilon=(1-\varepsilon)L+\varepsilon E\).
If \(x<y\), then
\[
L^\varepsilon(y)-L^\varepsilon(x)\ge\varepsilon(y-x)>0,
\qquad
\|L^\varepsilon-L\|_\infty=\varepsilon\|E-L\|_\infty.
\]
Thus \(L^\varepsilon\in\mathcal{L}^\uparrow\) and the asserted density holds.

Let \(f\) be a continuous strictly increasing function with \(f(0)=0\), \(f(1)=1\).
The intermediate value theorem and strict monotonicity give a bijection of \([0,1]\).
For \(0<r\le1\), the compact nonempty set
\(K_r=\{(x,y)\in[0,1]^2:|x-y|\ge r\}\)
defines the separation modulus
\[
m_f(r):=\min_{K_r}|f(x)-f(y)|
=\min_{0\le x\le1-r}\bigl(f(x+r)-f(x)\bigr)>0.
\]
Indeed, the minimum is attained, and a zero minimum would contradict injectivity.
The second minimum agrees with the first because strict increase
makes every increment over a longer interval larger than the
increment over its initial subinterval of length \(r\).
If \(|s-t|<m_f(r)\), then \(|f^{-1}(s)-f^{-1}(t)|<r\).
Thus, for an inverse tolerance \(0<\varepsilon\le1\), the choice
\(\delta=m_f(\varepsilon)\) proves uniform continuity directly.
For \(\varepsilon>1\), the conclusion is automatic since the inverse
takes values in \([0,1]\). The threshold \(m_f(r)\) is exact:
an interval attaining the minimum has inverse endpoint separation
\(r\), so no larger threshold ensures a separation strictly below \(r\).

For \(L\in\mathcal{L}^\uparrow\), set \(F=L^{-1}\).
Given \(s,t,\lambda\in[0,1]\), convexity yields
\[
L\bigl(\lambda F(s)+(1-\lambda)F(t)\bigr)
\le\lambda s+(1-\lambda)t.
\]
Applying the increasing map \(F\) proves concavity of \(F\), so \(F\in Y\).
Conversely, for \(F\in Y\) and \(L=F^{-1}\), concavity gives
\[
F\bigl(\lambda L(x)+(1-\lambda)L(y)\bigr)
\ge\lambda x+(1-\lambda)y;
\]
applying \(L\) proves convexity of \(L\). Therefore inversion is a bijection between the stated domains.

For any continuous \(g:[0,1]\to\mathbb R\), write
\(\omega_g(r):=\sup\{|g(s)-g(t)|:s,t\in[0,1],\ |s-t|\le r\}\).
Uniform continuity gives \(\omega_g(r)\to0\) as \(r\downarrow0\).
If \(L_n,L\in\mathcal{L}^\uparrow\), put \(\delta_n=\|L_n-L\|_\infty\).
For \(x_n=L_n^{-1}(t)\),
\[
|L(x_n)-t|=|L(x_n)-L_n(x_n)|\le\delta_n,
\qquad
|L_n^{-1}(t)-L^{-1}(t)|\le\omega_{L^{-1}}(\delta_n).
\]
Consequently \(\|L_n^{-1}-L^{-1}\|_\infty\le\omega_{L^{-1}}(\delta_n)\to0\) when \(\delta_n\to0\).
The identical estimate, with \(F_n,F\in Y\), proves continuity of \(F\mapsto F^{-1}\).
The dual mixture is therefore well-defined by convexity of \(Y\).

Finally, suppose that all curves in \(\mathcal{L}^\uparrow\) are indifferent.
For arbitrary \(L,M\in\mathcal{L}\), the curves \(L^\varepsilon,M^\varepsilon\), \(0<\varepsilon<1\), belong to \(\mathcal{L}^\uparrow\) and approach \(L,M\) uniformly as \(\varepsilon\downarrow0\).
Both \(L^\varepsilon\succsim M^\varepsilon\) and \(M^\varepsilon\succsim L^\varepsilon\) hold.
Closedness of the graph gives both \(L\succsim M\) and \(M\succsim L\).
Thus the relation on \(\mathcal{L}\) would be trivial, proving the contrapositive.
\end{proof}

\begin{lemma}[Convexity from polygonal slopes]
\label{lem:polygonal-slope-criterion}
Let \(P\) be continuous and piecewise affine on a compact interval, with nondecreasing successive slopes.
Then \(P\) is convex. If all slopes are nonnegative, \(P\) is nondecreasing; if all are positive, \(P\) is strictly increasing.
\end{lemma}
\begin{proof}
Let \(v\) be the step function equal to the slope on each open affine piece, with arbitrary values at the finitely many knots.
For \(x<y\), summing affine increments gives \(P(y)-P(x)=\int_x^y v(t)\,dt\).
For \(x<y<z\), ordered slopes give \(v(t)\ge v(r)\) for almost every \((r,t)\in(x,y)\times(y,z)\).
Therefore
\[
\begin{aligned}
0&\le\int_x^y\int_y^z(v(t)-v(r))\,dt\,dr\\
&=(y-x)(P(z)-P(y))-(z-y)(P(y)-P(x)).
\end{aligned}
\]
Rearrangement gives
\(P(y)\le\frac{z-y}{z-x}P(x)+\frac{y-x}{z-x}P(z)\), the convexity inequality for every interior point \(y\) of \([x,z]\).
The endpoint cases are equalities.
Nonnegative slopes give nonnegative increments by the same integral formula.
If every slope is positive, their finite minimum \(m\) is positive, and \(P(y)-P(x)\ge m(y-x)>0\).
\end{proof}

\begin{lemma}[Coordinate form of an affine function]
\label{lem:finite-affine}
Let \(D\subset\mathbb R^m\) be nonempty, open, and convex.
If \(f:D\to\mathbb R\) is affine, there exist \(a_0\in\mathbb R\), \(a\in\mathbb R^m\) with \(f(x)=a_0+a\cdot x\) for all \(x\in D\).
\end{lemma}

\begin{proof}
Choose \(p\in D\) and \(r>0\) with \(B_2(p,r)\subset D\).
Let \(e_i\) be the coordinate vectors, choose any \(h\in(0,r)\), and set
\(a_i=(f(p+he_i)-f(p))/h\).
Affinity on the segment from \(p\) to \(p+he_i\) gives
\(f(p+te_i)=f(p)+a_it\) for \(0\le t\le h\).
For \(h<t<r\), write \(p+he_i=(1-h/t)p+(h/t)(p+te_i)\) and rearrange the affine identity to obtain the same formula.
For \(-r<t<0\), use
\(f(p+te_i)+f(p-te_i)=2f(p)\).
Hence the coordinate formula holds for \(|t|<r\).

Whenever \(p,p+u,p+v,p+u+v\in D\), equality of the two midpoints gives
\(f(p+u+v)+f(p)=f(p+u)+f(p+v)\).
If \(\|w\|_2<r\), every coordinate partial sum of \(w\) has Euclidean norm at most \(\|w\|_2<r\).
Repeated application of this identity therefore yields
\(f(p+w)=f(p)+\sum_{i=1}^m a_iw_i\).
For arbitrary \(q\in D\), choose
any \(t\in(0,1]\) satisfying \(t\|q-p\|_2<r\).
Such a \(t\) exists, and convexity and affinity give
\((1-t)f(p)+tf(q)=f(p+t(q-p))=f(p)+t\,a\cdot(q-p)\).
Dividing by \(t>0\) gives the assertion with \(a_0=f(p)-a\cdot p\).
\end{proof}

\clearpage
\section{Parameter integration}
\label{app:parameter-integration}

\begin{lemma}[Parameter integration]
\label{lem:parameter-integration}
Let \(J\) be a nondegenerate compact interval and let
\(K_j:[0,1]\times J\to\mathbb R\), \(j\ge0\), be continuous,
with \(\partial_xK_j=K_{j+1}\) on the interior in \(x\).
Then
\[
I_j(x):=\int_0^1K_j(r,x)\,dr,\qquad
I_j'=I_{j+1},\qquad I_0^{(j)}=I_j\quad(j\ge0).
\]
The derivatives extend continuously to the endpoints of \(J\).
\end{lemma}
\begin{proof}
For \(j\ge0\), put
\[
\omega_j(d):=
\max_{\substack{r\in[0,1],\ x,z\in J\\|x-z|\le d}}
|K_j(r,x)-K_j(r,z)|.
\]
Compactness and continuity give uniform continuity of each \(K_j\), hence
\[
0\le\omega_j(d)<\infty\quad(d\ge0),\qquad
\lim_{d\downarrow0}\omega_j(d)=0.
\]
For \(u\ne0\), \(x,x+u\in J\),
\[
\begin{aligned}
\frac{I_j(x+u)-I_j(x)}u-I_{j+1}(x)
&=\int_0^1\int_0^1
\bigl(K_{j+1}(r,x+\theta u)-K_{j+1}(r,x)\bigr)
\,d\theta\,dr,\\
\left|\frac{I_j(x+u)-I_j(x)}u-I_{j+1}(x)\right|
&\le\int_0^1\omega_{j+1}(\theta|u|)\,d\theta
\le\omega_{j+1}(|u|)\longrightarrow0,\\
|I_j(x)-I_j(z)|&\le\omega_j(|x-z|)\longrightarrow0
\quad(z\to x).
\end{aligned}
\]
Thus \(I_j\) is continuous, \(I_j'=I_{j+1}\), and induction
gives \(I_0^{(j)}=I_j\). The same inequalities hold for the
one-sided endpoint increments.
\end{proof}

\clearpage
\section{Shape-preserving smooth density}
\label{app:bernstein-density}

\begin{lemma}[Shape-preserving smooth density]
\label{lem:bernstein-density}
For \(L\in\mathcal{L}\) and \(n\ge2\), define
\[
B_nL(x):=\sum_{j=0}^nL(j/n)b_{n,j}(x),\qquad
b_{n,j}(x):=\binom njx^j(1-x)^{n-j}.
\]
Then \(B_nL\in\mathcal{L}\cap C^\infty([0,1])\), and, with
\[
\omega_L(\delta):=
\max_{\substack{x,y\in[0,1]\\|x-y|\le\delta}}|L(x)-L(y)|,
\]
one has
\begin{equation}
\|B_nL-L\|_\infty
\le\inf_{\delta>0}
\left\{\omega_L(\delta)+\frac1{4n\delta^2}\right\},\qquad
\|B_nL-L\|_\infty\longrightarrow0.
\label{eq:bernstein-uniform-bound}
\end{equation}
\end{lemma}
\begin{proof}
Set \(b_{m,j}=0\) for \(j\notin\{0,\ldots,m\}\). The binomial
identities give
\[
b_{n,j}\ge0,\qquad \sum_{j=0}^nb_{n,j}(x)=(x+1-x)^n=1.
\]
At the endpoints,
\[
b_{n,j}(0)=\mathbf1_{\{j=0\}},\qquad
b_{n,j}(1)=\mathbf1_{\{j=n\}}.
\]
Consequently, \(B_nL(0)=L(0)=0\) and \(B_nL(1)=L(1)=1\).
For the derivatives, the coefficient identities
\[
j\binom nj=n\binom{n-1}{j-1},\qquad
(n-j)\binom nj=n\binom{n-1}{j}
\]
give
\[
b_{n,j}'=n(b_{n-1,j-1}-b_{n-1,j}).
\]
Thus
\[
\begin{aligned}
(B_nL)'
&=n\sum_{j=0}^nL(j/n)(b_{n-1,j-1}-b_{n-1,j})\\
&=n\sum_{j=0}^{n-1}
\bigl(L((j+1)/n)-L(j/n)\bigr)b_{n-1,j}\ge0,\\
(B_nL)''
&=n(n-1)\sum_{j=0}^{n-2}
\bigl(L((j+2)/n)-2L((j+1)/n)+L(j/n)\bigr)b_{n-2,j}\ge0.
\end{aligned}
\]
The last inequality uses
\[
2L((j+1)/n)\le L(j/n)+L((j+2)/n)
\qquad(0\le j\le n-2).
\]
For \(P\in C^2([0,1])\) and \(0\le a<b<c\le1\),
\begin{equation}
\begin{aligned}
&\frac{c-b}{c-a}P(a)+\frac{b-a}{c-a}P(c)-P(b)\\
&\quad=\frac{(b-a)(P(c)-P(b))-(c-b)(P(b)-P(a))}{c-a}\\
&\quad=\frac1{c-a}\int_a^b\int_b^c
(P'(t)-P'(s))\,dt\,ds\\
&\quad=\frac1{c-a}\int_a^b\int_b^c\int_s^tP''(r)\,dr\,dt\,ds.
\end{aligned}
\label{eq:curvature-shape}
\end{equation}
Hence \(P''\ge0\) implies convexity, and \(P''\le0\) implies
concavity. Also
\[
P(y)-P(x)=\int_x^yP'(r)\,dr
\begin{cases}
\ge0,&P'\ge0,\\
\ge(y-x)\min P'>0,&\min P'>0,
\end{cases}
\quad(x<y).
\]
Applied to \(P=B_nL\), these identities prove
\(B_nL\in\mathcal{L}\cap C^\infty([0,1])\).

For the approximation,
\[
\begin{aligned}
\sum_{j=0}^njb_{n,j}(x)
&=nx\sum_{j=0}^{n-1}b_{n-1,j}(x)=nx,\\
\sum_{j=0}^nj(j-1)b_{n,j}(x)
&=n(n-1)x^2\sum_{j=0}^{n-2}b_{n-2,j}(x)=n(n-1)x^2,\\
\sum_{j=0}^n(j/n-x)^2b_{n,j}(x)
&=\frac{n(n-1)x^2+nx}{n^2}-2x\frac{nx}{n}+x^2\\
&=\frac{x(1-x)}n,\\
\sum_{|j/n-x|>\delta}b_{n,j}(x)
&\le\frac1{\delta^2}\sum_{j=0}^n(j/n-x)^2b_{n,j}(x)
=\frac{x(1-x)}{n\delta^2}.
\end{aligned}
\]
Since \(0\le L\le1\),
\[
\begin{aligned}
|B_nL(x)-L(x)|
&\le\sum_{j=0}^n|L(j/n)-L(x)|b_{n,j}(x)\\
&\le\omega_L(\delta)\sum_{|j/n-x|\le\delta}b_{n,j}(x)
+\sum_{|j/n-x|>\delta}b_{n,j}(x)\\
&\le\omega_L(\delta)+\frac{x(1-x)}{n\delta^2}
\le\omega_L(\delta)+\frac1{4n\delta^2}.
\end{aligned}
\]
Here \(\max_{x\in[0,1]}x(1-x)=1/4\), attained at \(x=1/2\),
so the constant \(1/4\) is the exact uniform variance bound.
The cutoff \(\delta>0\) is arbitrary; taking its infimum proves
the estimate in \eqref{eq:bernstein-uniform-bound}.
The compactness argument for the moduli in
Lemma~\ref{lem:parameter-integration}, applied to
\(K(r,x)=L(x)\), gives \(\omega_L(\delta)\to0\).
For a prescribed \(\varepsilon>0\), choose any \(\delta>0\)
with \(\omega_L(\delta)<\varepsilon\). Then
\[
n\ge2,\qquad
n>\frac1{4\delta^2(\varepsilon-\omega_L(\delta))}
\ \Longrightarrow\ 
\|B_nL-L\|_\infty<\varepsilon.
\]
\end{proof}

\clearpage
\printbibliography[title={References}]
\end{document}